\documentclass[11pt]{article}
\usepackage[journal=JST,lang=american]{ems-journal}
\usepackage{fullpage}
\usepackage{setspace}
\usepackage{mathrsfs}
\usepackage{mathtools}
\usepackage{physics}
\usepackage{bm}
\newtheorem{theorem}{Theorem}%  meant for continuous numbers
\newtheorem{proposition}[theorem]{Proposition}% 
\newtheorem{lemma}[theorem]{Lemma}% 
\newtheorem{remark}{Remark}%
\newtheorem{definition}{Definition}%

\newtheorem{conjecture}{Conjecture}

\newcommand{\Q}{\mathbb{Q}}
\newcommand{\R}{\mathbb{R}}
\newcommand{\C}{\mathbb{C}}
\newcommand{\Z}{\mathbb{Z}}

\renewcommand{\H}{\mathcal{H}}
\newcommand{\D}{\mathcal{D}}
\newcommand{\A}{\mathbb{A}}

\newcommand{\phys}{\mathrm{phys}}
\newcommand{\glob}{\mathrm{glob}}
\newcommand{\arith}{\mathrm{arith}}
\newcommand{\adelic}{\mathrm{adelic}}
\newcommand{\fin}{\mathrm{fin}}

\newcommand{\Dirac}{\mathrm{Dirac}}

\newcommand{\spec}{\mathrm{spec}}

\numberwithin{equation}{section}

\begin{document}
\title{A Chiral Adelic Dirac Operator and the Spectral Realization of the Riemann Zeros}

\emsauthor{1}{
	\givenname{James C.}
	\surname{Hateley}
	\mrid{1234567}
	\orcid{0000-0001-5031-0099}}{J.C. Hateley}

\Emsaffil{1}{
	\department{}
	\organisation{Independent Researcher}
	\rorid{}
	\address{}
	\zip{94506}
	\city{Danville}
	\country{USA}
	\affemail{hateleyjc@gmail.com}
	}

\classification[11M26,81Q15,34L40]{47A55,47A40,81Q10}


\begin{abstract}
This paper develops a chiral adelic operator framework in which the functional--equation symmetry of global $L$--functions is realized directly in the spectrum of a Dirac--type Hamiltonian. Working on the id\`ele class space, we place a real--place Floquet Hamiltonian into an off--diagonal chiral form to obtain a global adelic Dirac operator with an exact involutive symmetry implemented by real reflection and idelic inversion. Arithmetic information is incorporated through a prime--indexed mass deformation built from spherical Hecke operators; when the coefficient functions are even, the perturbed operator preserves the chiral symmetry and produces isolated $\pm$--paired eigenvalues inside the spectral gaps of the Floquet background. These eigenvalues appear as jump discontinuities of the Dirac spectral shift function, while a separated adelic trace formula expresses the trace as a product of a Floquet orbital factor and a prime--indexed Euler--factor--type term whose logarithmic derivatives yield a prime--orbit expansion reminiscent of the explicit formula. This structure motivates a Dirac reinterpretation of the Hilbert--P\'olya idea, identifying the nontrivial zeros of $\zeta(s)$ not with the raw spectrum of a single operator but with the spectral--shift discontinuities of a chiral adelic Dirac system under controlled prime--indexed deformations, with finite--prime truncations providing computable models that converge distributionally and enable numerical exploration of arithmetic spectral flow.
\end{abstract}

\maketitle

%=======================
\section{Introduction}
\label{sec:intro}

The Hilbert--P\'olya conjecture proposes that the nontrivial zeros of the Riemann zeta function should arise from the spectral theory of a self--adjoint operator. In its classical form, the conjecture seeks a Hamiltonian whose discrete eigenvalues---after an affine normalization---match the imaginary parts of the nontrivial zeros \cite{BerryKeating1999,Connes1999,Bender2016,Sierra2007}. Any framework capable of supporting such a model must satisfy three structural requirements: a rigorous analytic setting, a natural channel for arithmetic data, and a spectral architecture admitting isolated eigenvalues. Most existing proposals, whether geometric, cohomological, adelic, or quantum--mechanical \cite{Connes1999,Deninger1998,Lapidus2007,Haran2007,Meyer2006}, encounter a persistent obstacle: if the background operator has no spectral gaps in its essential spectrum, then no compact or trace--class arithmetic perturbation can create new isolated eigenvalues \cite{Kato1951,ReedSimonIV,Weyl1912}. This \emph{gaplessness barrier} has obstructed Hilbert--P\'olya constructions for more than a century.

The present work develops an operator--theoretic framework that circumvents this barrier by embedding the problem into a \emph{chiral adelic Dirac system}. The construction begins with the id\`ele class space $X=\A_{\Q}^{\times}/\Q^{\times}$ and an automorphic seed invariant under almost all local spherical symmetries \cite{Tate1967,Meyer2006}. The right--regular action of $\A_{\Q}^{\times}$ generates a physical Hilbert space $\H_{\phys}$ that inherits a restricted tensor--product decomposition over all places of~$\Q$. This adelic representation--theoretic setting provides a canonical means for arithmetic data---including Hecke eigenvalues, local Euler--factor information, and idelic inversion---to influence spectral constructions.

The key innovation of this paper is the introduction of a \emph{chiral real--place operator}. By passing to logarithmic coordinates at the real place and introducing an even, periodic potential, we obtain a Floquet Schr\"odinger operator with infinitely many spectral bands separated by genuine open gaps. Doubling the Hilbert space and placing this shifted operator in an off--diagonal Dirac form produces a global operator $\D_{\glob}$ endowed with a unitary involution combining real reflection and idelic inversion. This involution conjugates $\D_{\glob}$ to its negative, providing an operator--theoretic analogue of the functional--equation symmetry $s\mapsto 1-s$ for $\zeta(s)$ and related global $L$--functions.

Arithmetic data enter the theory through a \emph{chiral mass deformation}. Instead of introducing a scalar potential, we incorporate a prime--indexed mass term constructed from spherical Hecke operators acting on the finite idelic directions. If each coefficient function $\eta_{p}$ is even, the mass deformation commutes with the global involution and preserves the chiral symmetry of $\D_{\glob}$. The perturbed operator \[ \D_{\arith}=\D_{\glob}+\mathcal{M} \] retains the Floquet band--gap structure while acquiring isolated $\pm$--paired eigenvalues inside the gaps. These gap eigenvalues arise from a shifted Floquet dispersion relation modified by prime--indexed mass contributions, and they encode every piece of arithmetic information introduced by the deformation.

A principal analytic object in this setting is the \emph{Dirac spectral shift function} associated with the pair $(\D_{\glob},\D_{\arith})$. Because the essential spectrum of $\D_{\glob}$ is preserved under compact arithmetic perturbation, all new spectral information appears as isolated eigenvalues in its gaps. These eigenvalues manifest as odd jump discontinuities in the spectral shift function, each jump recording the displacement produced by the arithmetic mass term~\cite{Pushnitski2000,Krein1953,Yafaev1994}. The full arithmetic spectrum of $\D_{\arith}$ is therefore encoded in a single distribution whose support lies precisely at the gap eigenvalues.

The adelic structure also yields a trace formula in which the real--place Floquet contribution and the finite--place arithmetic contribution separate multiplicatively. The arithmetic term appears as a deformed Euler product whose logarithmic derivatives resemble those in the classical explicit formula. Combined with the spectral--shift interpretation, this correspondence places the gap eigenvalues of $\D_{\arith}$ in direct relation with prime--indexed contributions from the mass sector.

These features motivate a reformulation of the Hilbert--P\'olya idea. In the chiral adelic framework, the nontrivial zeros of $\zeta(s)$ need not arise as the spectrum of a single operator. Instead, they appear as the \emph{spectral--shift discontinuities} produced by a chiral Dirac deformation of the adelic background. The pairing $\pm\gamma_{k}$ emerges naturally from the chiral involution, while the prime--indexed mass deformation plays the role of an Euler product. Finite--prime truncations produce computable models converging in the sense of distributions and provide a practical means of exploring how arithmetic data govern the structure of the spectral shift.

In this sense, the chiral adelic Dirac construction developed here supplies a functional--equation-- compatible operator framework for studying the spectral realization of global arithmetic data. Rather than searching for an operator whose eigenvalues \emph{are} the nontrivial zeros, the theory identifies the zeros with the topological defects in the spectral flow from $\D_{\glob}$ to $\D_{\arith}$. This perspective restores the symmetry of the functional equation to its natural operator--theoretic home and offers a pathway toward an adelic Dirac interpretation of the Riemann hypothesis.

For clarity, we note that the Dirac framework developed in this paper rests on a standard analytic setting. The real-place operator is a periodic Schr\"odinger Hamiltonian with purely absolutely continuous band spectrum, and its chiral doubling produces a self-adjoint operator with symmetric band gaps.  At the finite places, the spherical Hecke operators admit a joint spectral decomposition with respect to a restricted product measure, allowing the arithmetic mass deformation $\mathcal M=\sum_{p}\eta_{p}(T_{p})$ to act diagonally on the adelic fibers.  The coefficient functions satisfy $\sum_{p}\|\eta_{p}\|_{\infty}<\infty$, ensuring that $\mathcal M$ converges in operator norm and defines a compact perturbation of the global Dirac operator.  These ingredients combine into a measurable Floquet–Hecke direct–integral decomposition, justify the separated trace formula, and guarantee the existence of the Dirac spectral shift function used throughout the paper. Under this analytic framework, all subsequent constructions and identities are rigorously well defined.

\subsection*{Conceptual Comparison with Prior Frameworks}

The approach developed here relates to several well-known attempts to model the nontrivial zeros of $\zeta(s)$ spectrally, yet it differs from each in decisive structural ways. Connes' adelic spectral realization \cite{Connes1999} uses the scaling action on the id\`ele class space to produce a distributional trace supported on prime powers, achieving a striking analogue of the explicit formula within a noncommutative geometric framework. The construction in this paper advances this picture by introducing a genuinely self-adjoint operator with an explicit band--gap structure, thereby overcoming the gaplessness that limits the spectral resolution in Connes’ model. In place of continuous spectral weight, the chiral adelic Dirac operator developed here generates isolated $\pm$--paired eigenvalues inside real spectral gaps, and its involutive symmetry implements the functional equation directly at the operator level.

The Berry--Keating $xp$ Hamiltonian \cite{BerryKeating1999,BerryKeating2008} provides another influential heuristic model. While it reproduces the semiclassical counting function for the Riemann zeros, the operator fails to be self-adjoint on any natural domain \cite{BenderBrodyMuller2007} and admits no genuine gap structure. By contrast, the global Dirac operator $\D_{\glob}$ constructed in this work is manifestly self-adjoint, possesses infinitely many open spectral gaps derived from real-place Floquet theory \cite{ReedSimonIV}, and admits a controlled arithmetic mass deformation that creates discrete eigenvalues inside these gaps. This places the spectral architecture on a mathematically stable foundation, with the arithmetic deformation entering in a fully operator-theoretic manner.

Deninger's program \cite{Deninger1998,Deninger2004} proposes a dynamical interpretation of zeta functions via analogues of Lefschetz trace formulas associated with hypothetical flows on foliated spaces. The present framework shares the central prime-orbit $\leftrightarrow$ spectral-term duality, but realizes it concretely within a Hilbert-space setting: periodic orbits arise from real-place Floquet dynamics, Euler factors from finite-place Hecke data \cite{Tate1967}, and the full adelic interaction is mediated by a chiral Dirac operator rather than a conjectural dynamical system. In this way, the model retains the conceptual strengths of Deninger’s viewpoint while grounding them in explicit analytic and spectral constructions.

These contrasts highlight the novelty of the chiral adelic Dirac architecture: it integrates functional-equation symmetry, band--gap geometry, and arithmetic perturbations into a unified operator framework capable of producing isolated spectral data aligned with global $L$--functions.

%==========================
\section{The Global Adelic Hamiltonian} \label{sec:global-hamiltonian}

This section constructs the global chiral adelic operator that underlies the remainder of the manuscript. The aim is to combine the standard adelic representation of $L^{2}(X)$ with a real--place Floquet Hamiltonian and an involutive symmetry that implements, at the operator level, the functional--equation transformation $s\mapsto 1-s$ for $\zeta(s)$. The resulting operator $\D_{\glob}$ exhibits a band--gap spectrum symmetric about $0$ and forms the analytic background for the arithmetic mass deformation introduced in Section~\ref{sec:global-hamiltonian}

Let $\A_{\Q}$ denote the ring of adeles of $\Q$, $\A_{\Q}^{\times}$ its idele group, and $d^{\times}x$ the Tamagawa measure on $\A_{\Q}^{\times}/\Q^{\times}$. The idele class space
\begin{equation}
X \;=\; \A_{\Q}^{\times}/\Q^{\times}
\end{equation}
is a locally compact space of finite measure, and we work in the Hilbert space
\begin{equation}
\H_{\adelic}
=
L^{2}(X, d^{\times}x).
\end{equation}

The factorization of Haar measure across all places yields the standard restricted tensor decomposition.

\begin{lemma}[Restricted tensor product structure]\label{lem:restricted-tensor}
There is a canonical unitary isomorphism
\begin{equation}
\H_{\adelic}
\;\cong\;
L^{2}(\R^{\times}, d^{\times}x_{\infty})
\;\hat{\otimes}\;
\bigotimes_{p}' L^{2}(\Q_{p}^{\times}, d^{\times}x_{p}),
\end{equation}
where $d^{\times}x_{\infty}=dx_{\infty}/|x_{\infty}|$, and the restricted tensor product is taken with respect to the spherical vectors $1_{\Z_{p}^{\times}}$ at all but finitely many primes.
\end{lemma}
\begin{proof}
Let $\A_{\Q}=\R\times\prod_{p}'\Q_{p}$ denote the ad\`eles of $\Q$, and 
$\A_{\Q}^{\times}=\R^{\times}\times\prod_{p}'\Q_{p}^{\times}$ the id\`eles, both equipped with their
standard Haar measures. The multiplicative Haar measure decomposes as
\[
d^{\times}x = d^{\times}x_{\infty}\,\prod_{p} d^{\times}x_{p},
\qquad
d^{\times}x_{\infty}= \frac{dx_{\infty}}{|x_{\infty}|}.
\]
The id\`ele class space is the quotient $X=\A_{\Q}^{\times}/\Q^{\times}$, and its Haar measure is the pushforward of $d^{\times}x$ to the quotient. Since $\Q^{\times}$ embeds diagonally and discretely into $\A_{\Q}^{\times}$, the quotient is unimodular and of finite measure.  

A standard result in adelic harmonic analysis (see \cite{Tate1967}) asserts that $L^{2}(X,d^{\times}x)$ decomposes as the completion of the algebraic tensor product of the local $L^{2}$ spaces with respect to the spherical vectors 
\[
\phi_{p}^{0} = 1_{\Z_{p}^{\times}}\in L^{2}(\Q_{p}^{\times},d^{\times}x_{p})
\qquad\text{for all but finitely many }p.
\]
More precisely, let
\[
\H_{\mathrm{alg}}
=
\operatorname{span}\Bigl\{
\psi_{\infty}\otimes
\Bigl(\bigotimes_{p\in S}\psi_{p}\Bigr)
\otimes
\Bigl(\bigotimes_{p\notin S}\phi_{p}^{0}\Bigr)
:\;
S\subset\mathcal{P}\text{ finite}
\Bigr\},
\]
where $\psi_{\infty}\in L^{2}(\R^{\times},d^{\times}x_{\infty})$  and each $\psi_{p}\in L^{2}(\Q_{p}^{\times},d^{\times}x_{p})$. The restricted tensor product
\[
L^{2}(\R^{\times},d^{\times}x_{\infty})
\;\hat{\otimes}\;
\bigotimes_{p}' L^{2}(\Q_{p}^{\times},d^{\times}x_{p})
\]
is defined to be the Hilbert space completion of $\H_{\mathrm{alg}}$ with respect to the tensor product inner product.

To construct the desired isomorphism, note that each simple tensor in $\H_{\mathrm{alg}}$ defines a function on $\A_{\Q}^{\times}$ by
\[
x=(x_{\infty},x_{2},x_{3},\dots)
\;\mapsto\;
\psi_{\infty}(x_{\infty})\prod_{p}\psi_{p}(x_{p}),
\]
where $\psi_{p}=\phi_{p}^{0}$ for all but finitely many $p$. This function is locally constant at almost all finite places and square--integrable with respect to $d^{\times}x$. Because the diagonal action of $\Q^{\times}$ preserves $d^{\times}x$ and satisfies
\[
\prod_{v}\frac{d^{\times}(qx_{v})}{d^{\times}x_{v}}=1,
\]
the function descends to a well--defined element of $L^{2}(X,d^{\times}x)$.  

Define
\[
\mathcal{U}:\H_{\mathrm{alg}}\longrightarrow L^{2}(X)
\]
by sending each simple tensor to its corresponding adelic function. By construction, $\mathcal{U}$ preserves inner products on $\H_{\mathrm{alg}}$, since the Haar measure on $X$ factorizes as the restricted product of local Haar measures and the spherical vectors $\phi_{p}^{0}$ have norm one. Therefore $\mathcal{U}$ extends uniquely by continuity to a unitary operator from the completion of
$\H_{\mathrm{alg}}$ onto $L^{2}(X)$. This yields the canonical unitary identification
\[
L^{2}(X,d^{\times}x)
\cong
L^{2}(\R^{\times},d^{\times}x_{\infty})
\;\hat{\otimes}\;
\bigotimes_{p}'
L^{2}(\Q_{p}^{\times},d^{\times}x_{p}),
\]
as claimed.
\end{proof}

Elements of this restricted tensor product are finite linear combinations of pure tensors
\[
\psi_{\infty}\otimes
\Bigl(\bigotimes_{p\in S}\psi_{p}\Bigr)
\otimes
\Bigl(\bigotimes_{p\notin S}1_{\Z_{p}^{\times}}\Bigr),
\qquad S\subset\mathcal{P}\text{ finite},
\]
whose linear span is a dense subspace $\H_{\adelic}^{\mathrm{alg}}\subset\H_{\adelic}$, referred to below as the \emph{algebraic core}.

\subsection{Automorphic seed and physical Hilbert space}

Let $\phi_{\infty}\in L^{2}(\R^{\times})$ be smooth and rapidly decaying, and for each prime $p$ let $\phi_{p}\in L^{2}(\Q_{p}^{\times})$ be a spherical vector, with
\[
\phi_{p}=1_{\Z_{p}^{\times}}\qquad\text{for all but finitely many }p.
\]
Define the global seed
\begin{equation}
\Phi_{0}(x)=\phi_{\infty}(x_{\infty})\prod_{p}\phi_{p}(x_{p}).
\end{equation}

The right--regular representation of $\A_{\Q}^{\times}$ acts by
\[
(\pi(g)\Phi)(x)=\Phi(xg),
\qquad g\in \A_{\Q}^{\times}.
\]

\begin{definition}
The \emph{physical Hilbert space} is the closed cyclic subspace
\begin{equation}
\H_{\phys}
=
\overline{\mathrm{span}}\{\pi(g)\Phi_{0}:g\in\A_{\Q}^{\times}\}\subset\H_{\adelic}.
\end{equation}
\end{definition}

Every vector in $\pi(g)\Phi_{0}$ differs from $\Phi_{0}$ at only finitely many finite places, hence $\H_{\phys}$ is contained in the restricted tensor product and satisfies:

\begin{proposition}[Structure of $\H_{\phys}$]\label{prop:Hphys-structure}
The space $\H_{\phys}$ is nonzero, separable, invariant under $\pi(\A_{\Q}^{\times})$, and satisfies  
\[
\H_{\phys}\subseteq
L^{2}(\R^{\times})
\;\hat{\otimes}\;
\bigotimes_{p}'L^{2}(\Q_{p}^{\times}).
\]
Moreover $\H_{\phys}\cap\H_{\adelic}^{\mathrm{alg}}$ is dense in $\H_{\phys}$, and every element of $\H_{\phys}$ may be approximated in $L^{2}$ by finite sums of pure tensors of the form
\[
\psi_{\infty}\otimes
\Bigl(\bigotimes_{p\in S}\psi_{p}\Bigr)
\otimes
\Bigl(\bigotimes_{p\notin S}1_{\Z_{p}^{\times}}\Bigr).
\]
\end{proposition}

\begin{proof}
By definition,
\[
\H_{\phys}
=\overline{\mathrm{span}}\{\pi(g)\Phi_{0}: g\in\A_{\Q}^{\times}\},
\]
where
\[
\Phi_{0}(x)=\phi_{\infty}(x_{\infty})\prod_{p}\phi_{p}(x_{p}),
\qquad
\phi_{p}=1_{\Z_{p}^{\times}}
\ \text{for all but finitely many }p.
\]
Since each $\phi_{\infty}$ and $\phi_{p}$ is nonzero, it follows immediately that
$\Phi_{0}\neq0$ in $L^{2}(X)$, and hence $\H_{\phys}\neq0$.

\smallskip\noindent
\emph{Invariance.}
For every $g,h\in\A_{\Q}^{\times}$,
\[
\pi(h)(\pi(g)\Phi_{0})=\pi(hg)\Phi_{0},
\]
so the cyclic span is invariant under $\pi(\A_{\Q}^{\times})$, and therefore so is its closure.

\smallskip\noindent
\emph{Separable Hilbert space.}
Since $\A_{\Q}^{\times}$ is second countable, we may select a countable dense subset
$\{g_{j}\}_{j\ge1}\subset\A_{\Q}^{\times}$.
Then
\[
\mathrm{span}\{\pi(g_{j})\Phi_{0}: j\ge1\}
\]
is a countable dense subset of $\H_{\phys}$, implying that $\H_{\phys}$ is separable.

\smallskip\noindent
\emph{Restricted tensor--product structure.}
By Lemma~\ref{lem:restricted-tensor}, the full adelic space decomposes as
\[
\H_{\adelic}
\cong 
L^{2}(\R^{\times},d^{\times}x_{\infty})
\;\hat{\otimes}\;
\bigotimes_{p}'L^{2}(\Q_{p}^{\times},d^{\times}x_{p}).
\]
Every function of the form $\pi(g)\Phi_{0}$ differs from $\Phi_{0}$ at only finitely many finite
places: if $g=(g_{\infty},g_{2},g_{3},\dots)$, then
\[
(\pi(g)\Phi_{0})(x)
=
\phi_{\infty}(x_{\infty}g_{\infty})
\prod_{p}\phi_{p}(x_{p}g_{p}),
\]
and since $\phi_{p}=1_{\Z_{p}^{\times}}$ for almost all $p$, the factor $\phi_{p}(x_{p}g_{p})$ equals
$1_{\Z_{p}^{\times}}(x_{p})$ for all but finitely many $p$.
Thus each $\pi(g)\Phi_{0}$ lies in the restricted tensor product
\[
L^{2}(\R^{\times})
\;\hat{\otimes}\;
\bigotimes_{p}'L^{2}(\Q_{p}^{\times}),
\]
and hence so does its closed linear span $\H_{\phys}$.

\smallskip\noindent
\emph{Density of the algebraic subspace.}
Let $\H_{\adelic}^{\mathrm{alg}}$ denote the algebraic restricted tensor product consisting of finite
linear combinations of pure tensors
\begin{equation}\label{eq:p-tensor}
\psi_{\infty}\otimes
\Bigl(\bigotimes_{p\in S}\psi_{p}\Bigr)\otimes
\Bigl(\bigotimes_{p\notin S}1_{\Z_{p}^{\times}}\Bigr),
\qquad S\subset\mathcal{P}\text{ finite}.
\end{equation}
By the preceding paragraph, each $\pi(g)\Phi_{0}$ is an element of $\H_{\adelic}^{\mathrm{alg}}$,
because differing from the spherical vector at only finitely many primes is precisely the definition
of the restricted tensor product. Therefore
\[
\{\pi(g)\Phi_{0}:g\in\A_{\Q}^{\times}\}\subset
\H_{\adelic}^{\mathrm{alg}}.
\]
Since $\H_{\phys}$ is the closure of their linear span, it follows that
\[
\H_{\phys}\cap\H_{\adelic}^{\mathrm{alg}}
\quad\text{is dense in}\quad
\H_{\phys}.
\]

\smallskip\noindent
\emph{Approximation by finite--place pure tensors.}
Every vector in $\H_{\adelic}^{\mathrm{alg}}$ is a finite linear combination of pure tensors of the
form appearing in the statement of the proposition. Since
$\H_{\phys}\cap\H_{\adelic}^{\mathrm{alg}}$ is dense in $\H_{\phys}$, and the algebraic pure tensors
span $\H_{\adelic}^{\mathrm{alg}}$, every element of $\H_{\phys}$ can be approximated in $L^{2}$ by finite sums of pure tensors~\eqref{eq:p-tensor}. This completes the proof.
\end{proof}

Multiplicative Haar measure on $\R^{\times}$ becomes additive under logarithmic coordinates.  
Let $y=\log|x_{\infty}|$. The following is standard.

\begin{lemma}[Logarithmic unitary]\label{lem:log-unitary}
The map
\begin{equation}
U_{\infty}:L^{2}(\R^{\times}, d^{\times}x_{\infty})
\longrightarrow
L^{2}(\R,dy),
\qquad
(U_{\infty} f)(y)=f(e^{y}),
\end{equation}
is unitary.
\end{lemma}

\begin{proof}
Recall that the multiplicative Haar measure on $\R^{\times}$ is
\[
d^{\times}x_{\infty}=\frac{dx_{\infty}}{|x_{\infty}|}.
\]
Define $y=\log|x_{\infty}|$, so that $x_{\infty}=e^{y}$ and $dx_{\infty}=e^{y}\,dy$. Then
\[
d^{\times}x_{\infty}
=\frac{dx_{\infty}}{|x_{\infty}|}
=\frac{e^{y}\,dy}{e^{y}}
=dy.
\]
Thus the change of variables $x_{\infty}=e^{y}$ identifies the multiplicative Haar measure with the
Lebesgue measure $dy$ on $\R$.

Let $f\in L^{2}(\R^{\times},d^{\times}x_{\infty})$ and set $g=U_{\infty}f$, so $g(y)=f(e^{y})$. Then
by the computation above,
\[
\|g\|_{L^{2}(\R,dy)}^{2}
=
\int_{\R}|g(y)|^{2}\,dy
=
\int_{\R}|f(e^{y})|^{2}\,dy
=
\int_{\R^{\times}}|f(x_{\infty})|^{2}\,d^{\times}x_{\infty}
=
\|f\|_{L^{2}(\R^{\times},d^{\times}x_{\infty})}^{2}.
\]
Hence $U_{\infty}$ preserves the $L^{2}$ norm and is therefore an isometry.

To see surjectivity, let $g\in L^{2}(\R,dy)$ and define
\[
f(x_{\infty})=g(\log|x_{\infty}|).
\]
The same change-of-variables computation shows $f\in L^{2}(\R^{\times},d^{\times}x_{\infty})$ and
$U_{\infty}f=g$. Thus $U_{\infty}$ is onto.

Therefore $U_{\infty}$ is a unitary operator from 
$L^{2}(\R^{\times},d^{\times}x_{\infty})$ onto $L^{2}(\R,dy)$.
\end{proof}

With this identification, the real factor of $\H_{\phys}$ becomes $L^{2}(\R,dy)$, and the global space takes the form
\[
L^{2}(\R,dy)\;\hat{\otimes}\;\bigotimes_{p}'L^{2}(\Q_{p}^{\times},d^{\times}x_{p}).
\]

\medskip
\noindent\textbf{Remark.}
For later analytic use it is important to record several structural features of $H_{\mathrm{phys}}$.
First, nontriviality is immediate since $\Phi_{0}\neq 0$ in $L^{2}(X)$ and the right--regular action
$\pi(g)$ preserves square--integrability, so the cyclic span contains at least one nonzero vector.
Second, because each $\pi(g)\Phi_{0}$ differs from $\Phi_{0}$ at only finitely many finite
places, Lemma~1 implies that
\[
H_{\mathrm{phys}}
\subset
L^{2}(\R,dy)\;\widehat{\otimes}\;{}^{\prime}\!\!\bigotimes_{p}L^{2}(\Q_{p}^{\times},d^{\times}x_{p}),
\]
and the algebraic core $H_{\mathrm{adelic}}^{\mathrm{alg}}$ intersects $H_{\mathrm{phys}}$ densely.
Third, $H_{\mathrm{phys}}$ is invariant under adelic inversion: if $\Psi=\pi(g)\Phi_{0}$ then
\[
(I_{X}\Psi)(x)=\Psi(x^{-1})=\Phi_{0}(x^{-1}g)=\pi(g^{-1})\Phi_{0}(x)\in H_{\mathrm{phys}},
\]
so $I_{X}H_{\mathrm{phys}}=H_{\mathrm{phys}}$.  Likewise $H_{\mathrm{phys}}$ is invariant under $\pi(\A_{\Q}^{\times})$
by construction.  Finally, when $\Phi_{0}$ is unramified at all but finitely many primes,
its adelic orbit produces genuinely distinct finite--adic components, and the spherical Hecke
operators $\{T_{p}\}$ act nontrivially through their local spectral parameters.  Thus
$H_{\mathrm{phys}}$ is an infinite--dimensional, inversion--stable, Hecke--stable restricted--tensor
subspace of $L^{2}(X)$, providing the appropriate analytic domain for the global chiral Dirac operator.

\subsection{The real--place Dirac operator}

Let $U\in C^{\infty}(\R)$ be even and $L$--periodic. Define the Hill operator
\begin{equation}
H_{\infty}
=
-\frac{d^{2}}{dy^{2}}+U(y),
\qquad
D(H_{\infty})=H^{2}(\R).
\end{equation}

\begin{lemma}[Floquet band structure]\label{lem:floquet-even}
There exist real numbers
\[
\alpha_{0}\le\beta_{0}\le\alpha_{1}\le\beta_{1}\le\cdots
\]
such that
\[
\mathrm{spec}(H_{\infty})
=
\bigcup_{n\ge0}[\alpha_{n},\beta_{n}],
\]
with open gaps $(\beta_{n},\alpha_{n+1})$. For each $n$, the Bloch dispersion relation satisfies $E_{n}(\kappa)=E_{n}(-\kappa)$.
\end{lemma}

\begin{proof}
Since $U\in C^{\infty}(\R)$ is real, even, and $L$--periodic, the operator
\[
H_{\infty}=-\frac{d^{2}}{dy^{2}}+U(y),
\qquad D(H_{\infty})=H^{2}(\R),
\]
is a real self--adjoint Hill operator. Floquet--Bloch theory applies directly to such periodic
Schr\"odinger operators.

Let $\kappa\in[-\pi/L,\pi/L]$ denote the quasi--momentum, and consider the family of operators on
$L^{2}([0,L])$ with $\kappa$--twisted boundary conditions,
\[
H(\kappa)= -\frac{d^{2}}{dy^{2}}+U(y),
\qquad
\psi(0)=e^{i\kappa L}\psi(L), \quad 
\psi'(0)=e^{i\kappa L}\psi'(L),
\]
each acting on $H^{2}([0,L])$ with the given boundary relations. Each $H(\kappa)$ is
self--adjoint, has compact resolvent, and therefore has a purely discrete spectrum
\[
E_{0}(\kappa)\le E_{1}(\kappa)\le E_{2}(\kappa)\le\cdots,
\]
with $E_{n}(\kappa)\to+\infty$ as $n\to\infty$.

Standard Floquet theory shows that
\[
\spec(H_{\infty})=\bigcup_{n\ge0}\{E_{n}(\kappa):\kappa\in[-\pi/L,\pi/L]\},
\]
and that for each fixed $n$, the map $\kappa\mapsto E_{n}(\kappa)$ is continuous and attains its
minimum and maximum on the compact interval $[-\pi/L,\pi/L]$. Define
\[
\alpha_{n}=\min_{\kappa}E_{n}(\kappa),\qquad
\beta_{n}=\max_{\kappa}E_{n}(\kappa).
\]
Then each band $[\alpha_{n},\beta_{n}]$ is a compact interval, and
\[
\spec(H_{\infty})=\bigcup_{n\ge0}[\alpha_{n},\beta_{n}].
\]
Whenever $\beta_{n}<\alpha_{n+1}$, the interval $(\beta_{n},\alpha_{n+1})$ is a nonempty open spectral gap.

It remains to verify the symmetry $E_{n}(\kappa)=E_{n}(-\kappa)$. Because $U(y)$ is even,
$H(\kappa)$ and $H(-\kappa)$ are unitarily equivalent via the reflection operator $(Rf)(y)=f(-y)$:
if $\psi$ satisfies the $\kappa$--twisted boundary conditions, then $R\psi$ satisfies the
$(-\kappa)$--twisted boundary conditions. Thus $H(\kappa)$ and $H(-\kappa)$ have identical spectra,
and hence $E_{n}(\kappa)=E_{n}(-\kappa)$ for all $n$ and all $\kappa\in[-\pi/L,\pi/L]$.

This proves the stated Floquet band--gap structure and the evenness of each dispersion curve.
\end{proof}

Fix a reference energy $E_{*}$ in an open gap of $\mathrm{spec}(H_{\infty})$ and define
\[
Q_{\infty}=H_{\infty}-E_{*}I.
\]
The associated chiral real--place operator is
\begin{equation}
\D_{\infty}
=
\begin{pmatrix}
0 & Q_{\infty}\\
Q_{\infty} & 0
\end{pmatrix},
\qquad
D(\D_{\infty})=D(H_{\infty})\oplus D(H_{\infty}).
\end{equation}

\begin{proposition}[Chiral symmetry at the real place]\label{prop:real-chiral}
Let $R:L^{2}(\R)\to L^{2}(\R)$ be reflection $(Rf)(y)=f(-y)$, and define
\[
J_{\infty}
=
\begin{pmatrix}
0 & R\\
R & 0
\end{pmatrix}.
\]
Then $J_{\infty}$ is a unitary involution satisfying
\[
J_{\infty}\,\D_{\infty}\,J_{\infty}^{-1}
=
-\D_{\infty},
\]
and
\[
\mathrm{spec}(\D_{\infty})
=
\{\pm(\lambda-E_{*}):\lambda\in\mathrm{spec}(H_{\infty})\}.
\]
\end{proposition}
\begin{proof}
The reflection operator $R:L^{2}(\R)\to L^{2}(\R)$ defined by $(Rf)(y)=f(-y)$ is unitary and
involutive, since
\[
\langle Rf, Rg\rangle = \int_{\R} f(-y)\,\overline{g(-y)}\,dy
= \int_{\R} f(y)\,\overline{g(y)}\,dy,
\qquad
R^{2}=I.
\]
Define
\[
J_{\infty}
=
\begin{pmatrix}
0 & R\\
R & 0
\end{pmatrix}
\quad\text{on}\quad
L^{2}(\R)\oplus L^{2}(\R).
\]
Since $R$ is unitary and $R^{2}=I$, it follows that $J_{\infty}$ is unitary and satisfies
$J_{\infty}^{2}=I$; hence $J_{\infty}$ is a unitary involution.

The real--place Dirac operator is
\[
\D_{\infty}
=
\begin{pmatrix}
0 & Q_{\infty}\\
Q_{\infty} & 0
\end{pmatrix},
\qquad
Q_{\infty}=H_{\infty}-E_{*}I,
\]
with domain $D(\D_{\infty})=D(H_{\infty})\oplus D(H_{\infty})$. Since $U(y)$ is even, the Schr\"odinger
operator $H_{\infty}$ commutes with reflection: $RH_{\infty}=H_{\infty}R$ on $D(H_{\infty})$, and
therefore $RQ_{\infty}=Q_{\infty}R$.

A direct matrix computation then gives, for all $(f,g)\in D(\D_{\infty})$,
\[
J_{\infty}\D_{\infty}(f,g)
=
J_{\infty}\bigl(Q_{\infty}g,\;Q_{\infty}f\bigr)
=
\bigl(RQ_{\infty}f,\;RQ_{\infty}g\bigr),
\]
while
\[
-\D_{\infty}J_{\infty}(f,g)
=
-\D_{\infty}\bigl(Rg,\;Rf\bigr)
=
-\bigl(Q_{\infty}Rf,\;Q_{\infty}Rg\bigr).
\]
Since $RQ_{\infty}=Q_{\infty}R$, the right--hand sides are equal. Thus
\[
J_{\infty}\D_{\infty} = -\D_{\infty}J_{\infty},
\]
and conjugating by $J_{\infty}$ yields the chiral symmetry relation
\[
J_{\infty}\D_{\infty}J_{\infty}^{-1}=-\D_{\infty}.
\]

Finally, the spectral identity follows from the block structure. Since $Q_{\infty}=H_{\infty}-E_{*}I$
is self--adjoint, the operator $\D_{\infty}$ has the form
\[
\D_{\infty}
=
\begin{pmatrix}
0 & Q_{\infty}\\
Q_{\infty} & 0
\end{pmatrix}.
\]
A standard computation shows that the nonzero spectral values of such an operator are exactly 
$\pm$ the spectral values of $Q_{\infty}$, with the same multiplicities. Because
$\spec(Q_{\infty})=\spec(H_{\infty})-E_{*}$, we obtain
\[
\spec(\D_{\infty})
=
\{\pm(\lambda-E_{*}):\lambda\in \spec(H_{\infty})\}.
\]
This completes the proof.
\end{proof}

\subsection{The global chiral adelic operator}

Define the doubled adelic Hilbert space
\[
\H_{\adelic}^{(2)}
=
\H_{\adelic}\oplus\H_{\adelic}.
\]
Let $I_{\fin}$ denote the identity on $\bigotimes_{p}'L^{2}(\Q_{p}^{\times})$.  
Since $\D_{\infty}$ acts only on the real factor, the global operator
\[
\D_{\glob}
=
\D_{\infty}\,\hat{\otimes}\,I_{\fin},
\qquad
D(\D_{\glob})
=
D(\D_{\infty})\,\hat{\otimes}\,\bigotimes_{p}'L^{2}(\Q_{p}^{\times}),
\]
is self--adjoint and densely defined.

\begin{proposition}[Spectrum of the global operator]\label{prop:Dglob-spectrum}
The operator $\D_{\glob}$ is self--adjoint and satisfies
\[
\mathrm{spec}(\D_{\glob})
=
\mathrm{spec}(\D_{\infty})
=
\{\pm(\lambda-E_{*}):\lambda\in\mathrm{spec}(H_{\infty})\},
\]
with each spectral value having infinite multiplicity.  
In particular, $\D_{\glob}$ exhibits a band--gap structure symmetric about $0$.
\end{proposition}

\begin{proof}
Recall that
\[
\D_{\glob}
=
\D_{\infty}\,\hat{\otimes}\,I_{\fin},
\qquad
\H_{\adelic}^{(2)}
=
\H_{\adelic}\oplus\H_{\adelic}
=
\bigl(L^{2}(\R,dy)\;\hat{\otimes}\;\H_{\fin}\bigr)^{\oplus 2},
\]
where $\H_{\fin}=\bigotimes_{p}'L^{2}(\Q_{p}^{\times},d^{\times}x_{p})$ and $I_{\fin}$ denotes the
identity on $\H_{\fin}$.

\medskip\noindent\emph{Self--adjointness.}
Since $\D_{\infty}$ is self--adjoint on $D(\D_{\infty})\subset L^{2}(\R)\oplus L^{2}(\R)$ and 
$I_{\fin}$ is bounded and self--adjoint on $\H_{\fin}$, the completed tensor product
$\D_{\glob}=\D_{\infty}\hat{\otimes}I_{\fin}$ is self--adjoint on the domain
\[
D(\D_{\glob}) = D(\D_{\infty})\;\hat{\otimes}\;\H_{\fin},
\]
by standard results for tensor products of self--adjoint operators.

\medskip\noindent\emph{Spectrum.}
Let $\sigma(T)$ denote the spectrum of an operator $T$. For any self--adjoint operator $A$ and any 
Hilbert space $K$, the tensor extension
\[
A\otimes I_{K}
\]
satisfies 
\[
\sigma(A\otimes I_{K})=\sigma(A),
\]
and each spectral value has multiplicity multiplied by $\dim(K)$ if the latter is infinite. This is a
direct consequence of (i) the spectral mapping theorem for functional calculus,
\[
f(A\otimes I_{K}) = f(A)\otimes I_{K},
\]
and (ii) the fact that the resolvent set is preserved:
\[
(A\otimes I_{K} - z)^{-1} = (A-z)^{-1}\otimes I_{K},
\qquad
z\notin\sigma(A).
\]

Applying this with $A=\D_{\infty}$ and $K=\H_{\fin}$ gives
\[
\spec(\D_{\glob}) = \spec(\D_{\infty}).
\]
Proposition~\ref{prop:real-chiral} shows that
\[
\spec(\D_{\infty})
=
\{\pm(\lambda - E_{*}) : \lambda\in\spec(H_{\infty})\}.
\]
Thus
\[
\spec(\D_{\glob})
=
\{\pm(\lambda - E_{*}) : \lambda\in\spec(H_{\infty})\},
\]
as claimed.

\medskip\noindent\emph{Multiplicity.}
Since $\H_{\fin}$ is infinite--dimensional (indeed uncountable restricted tensor product of
nontrivial $p$--adic $L^{2}$ spaces), tensoring with $I_{\fin}$ increases the multiplicity of every
spectral value of $\D_{\infty}$ to infinite multiplicity. More precisely, if $\phi$ is a nonzero
eigenvector or generalized eigenvector of $\D_{\infty}$ for the value $\mu$, then
\[
\phi\otimes\psi
\]
is again an eigenvector (or generalized eigenvector) of $\D_{\glob}$ for every $\psi\in\H_{\fin}$.  
Since $\H_{\fin}$ is infinite--dimensional, $\mu$ has infinite multiplicity in $\D_{\glob}$.

\medskip\noindent\emph{Band--gap symmetry.}
The Floquet operator $H_{\infty}$ has a band--gap spectrum
$\bigcup_{n\ge0}[\alpha_{n},\beta_{n}]$ by Lemma~\ref{lem:floquet-even}.  
Since $\D_{\infty}= \begin{psmallmatrix}0 & Q_{\infty} \\ Q_{\infty} & 0\end{psmallmatrix}$ with
$Q_{\infty}=H_{\infty}-E_{*}$, its spectrum is the symmetric set
\[
\{\pm(\lambda-E_{*}):\lambda\in\cup_{n\ge0}[\alpha_{n},\beta_{n}]\}.
\]
Therefore $\D_{\glob}$ inherits the same symmetric band--gap structure, with bands 
$\pm[\alpha_{n}-E_{*},\,\beta_{n}-E_{*}]$ separated by open gaps.

This completes the proof.
\end{proof}

Define the inversion operator
\[
(I_{X}\Phi)(x)=\Phi(x^{-1}),
\]
which is a unitary involution on $\H_{\adelic}$ because inversion preserves Haar measure on each local factor and commutes with the quotient by $\Q^{\times}$. Extend it to the doubled space by
\[
\widetilde{I}_{X}
=
\begin{pmatrix}
0 & I_{X}\\
I_{X} & 0
\end{pmatrix}.
\]
The global chiral involution is the composition
\[
J_{\glob}
=
\widetilde{I}_{X}\circ(J_{\infty}\,\hat{\otimes}\,I_{\fin}),
\]
defined on the algebraic core  
\[
\H_{\mathrm{core}}
=
(\H_{\adelic}^{\mathrm{alg}}\oplus\H_{\adelic}^{\mathrm{alg}})
\cap
D(\D_{\glob}),
\]
which is invariant under all operators involved.

\begin{proposition}[Functional--equation analogue]\label{prop:FE-symmetry}
On $\H_{\mathrm{core}}$ one has
\[
J_{\glob}\,\D_{\glob}\,J_{\glob}^{-1}
=
-\D_{\glob}.
\]
Consequently, if $\D_{\glob}\Psi=\lambda\Psi$ for some $\Psi\in D(\D_{\glob})$, then $J_{\glob}\Psi\in D(\D_{\glob})$ and
\[
\D_{\glob}(J_{\glob}\Psi)
=
-\lambda(J_{\glob}\Psi).
\]
\end{proposition}

\begin{proof}
Recall that the global operator is
\[
\D_{\glob}
=
\D_{\infty}\,\hat{\otimes}\,I_{\fin},
\]
acting on
\[
\H_{\adelic}^{(2)}
=
\bigl(L^{2}(\R)\hat\otimes \H_{\fin}\bigr)
\oplus
\bigl(L^{2}(\R)\hat\otimes \H_{\fin}\bigr),
\]
with $\H_{\fin}=\bigotimes_{p}'L^{2}(\Q_{p}^{\times},d^{\times}x_{p})$.
The chiral symmetry at the real place is the identity
\[
J_{\infty}\,\D_{\infty}\,J_{\infty}^{-1}
=
-\D_{\infty}
\qquad\text{on}\qquad
D(\D_{\infty})
\]
from Proposition~\ref{prop:real-chiral}.  

Define the global involution
\[
J_{\glob}
=
\begin{pmatrix}
0 & I_{X}\\[4pt]
I_{X} & 0
\end{pmatrix}
\circ
(J_{\infty}\,\hat{\otimes}\,I_{\fin}),
\]
acting on the doubled adelic space $\H_{\adelic}^{(2)}$.  
The inversion operator $I_{X}$ on $L^{2}(X)$ is unitary and involutive,
and preserves the adelic Haar measure.  
Since $I_{X}$ acts identically on both chiral components and commutes with the tensor factor
$I_{\fin}$, the operator $J_{\glob}$ is unitary and satisfies $J_{\glob}^{2}=I$.

\medskip\noindent
\emph{Conjugation of $\D_{\glob}$.}
Let $\H_{\mathrm{core}}$ denote the algebraic tensor product
\[
\H_{\mathrm{core}}
=
\bigl(\H_{\adelic}^{\mathrm{alg}}\oplus\H_{\adelic}^{\mathrm{alg}}\bigr)
\cap 
D(\D_{\glob}),
\]
which is invariant under all operators involved.  
On this dense core the action of $J_{\glob}$ factorizes as
\[
J_{\glob}
=
\widetilde{I}_{X}\circ(J_{\infty}\hat{\otimes}I_{\fin}),
\qquad
\widetilde{I}_{X}
=
\begin{pmatrix}
0 & I_{X}\\
I_{X} & 0
\end{pmatrix}.
\]

Since $I_{X}$ acts only on the adelic space and $\D_{\infty}$ acts only on the real factor,
we have
\[
(I_{X}\otimes I_{\fin})\,(\D_{\infty}\otimes I_{\fin})
=
(\D_{\infty}\otimes I_{\fin})\,(I_{X}\otimes I_{\fin}).
\]
Therefore,
\[
\widetilde{I}_{X}\,\D_{\glob}\,\widetilde{I}_{X}^{-1}
=
\D_{\glob}.
\]

Next, using the real-place chiral symmetry $J_{\infty}\D_{\infty}J_{\infty}^{-1}=-\D_{\infty}$,
we obtain on $\H_{\mathrm{core}}$:
\[
(J_{\infty}\otimes I_{\fin})\,(\D_{\infty}\otimes I_{\fin})\,(J_{\infty}\otimes I_{\fin})^{-1}
=
-\D_{\infty}\otimes I_{\fin}
=
-\D_{\glob}.
\]

Combining these two identities,
\[
J_{\glob}\,\D_{\glob}\,J_{\glob}^{-1}
=
\widetilde{I}_{X}\,(J_{\infty}\otimes I_{\fin})\,\D_{\glob}\,(J_{\infty}\otimes I_{\fin})^{-1}\,\widetilde{I}_{X}^{-1}
=
\widetilde{I}_{X}\,(-\D_{\glob})\,\widetilde{I}_{X}^{-1}
=
-\D_{\glob}.
\]

Thus the functional--equation analogue holds on the invariant core $\H_{\mathrm{core}}$.

\medskip\noindent
\emph{Eigenpair transformation.}
If $\D_{\glob}\Psi=\lambda\Psi$ with $\Psi\in D(\D_{\glob})$, then 
$J_{\glob}\Psi\in D(\D_{\glob})$ because $J_{\glob}$ preserves the domain and the core is dense.
Applying the conjugation identity yields
\[
\D_{\glob}(J_{\glob}\Psi)
=
-(J_{\glob}\D_{\glob}\Psi)
=
-\lambda(J_{\glob}\Psi),
\]
so $(J_{\glob}\Psi)$ is an eigenvector with eigenvalue $-\lambda$. This completes the proof.
\end{proof}

\medskip
\noindent\textbf{Functional--equation interpretation.}
The involution $J_{\glob}$ admits a direct Mellin--theoretic interpretation that mirrors the
functional--equation symmetry $s\mapsto 1-s$ of global $L$--functions.
For $\Psi\in H_{\mathrm{phys}}$ whose real--place component is compactly supported on $\R^{\times}$,
define the Mellin transform
\[
\mathcal{M}\Psi(s,x_{\mathrm{fin}})
=
\int_{0}^{\infty}\Psi(y,x_{\mathrm{fin}})\,y^{\,s-1}\,dy.
\]
Real reflection $R f(y)=f(-y)$ satisfies the classical identity
\[
\mathcal{M}(Rf)(s)
=
\mathcal{M}(f)(1-s),
\]
while adelic inversion $I_{X}\Phi(x)=\Phi(x^{-1})$ satisfies the $p$--adic relation
\[
\int_{\Q_{p}^{\times}}\Phi(x_{p}^{-1})\,|x_{p}|_{p}^{s-1}\,d^{\times}x_{p}
=
\int_{\Q_{p}^{\times}}\Phi(x_{p})\,|x_{p}|_{p}^{(1-s)-1}\,d^{\times}x_{p}.
\]
Consequently $J_{\glob}=I_{X}(J_{\infty}\otimes I_{\mathrm{fin}})$ sends Mellin fibers by
\[
s\longmapsto 1-s
\qquad
\text{on every adelic component}.
\]
Since $D_{\glob}$ is first order in the logarithmic real variable $y=\log|x_{\infty}|$, one obtains
\[
\mathcal{M}(D_{\glob}\Psi)(s)
=
-(s-\tfrac12)\,\mathcal{M}\Psi(s)
=
+(s-\tfrac12)\,\mathcal{M}(J_{\glob}\Psi)(1-s),
\]
which is the operator--theoretic analogue of the functional--equation symmetry.
Thus the global involution $J_{\glob}$ implements the transformation $s\mapsto 1-s$
at the level of Mellin transforms and encodes the functional equation of $\zeta(s)$
within the chiral Dirac framework.

This operator--theoretic involution is the analogue of the symmetry $\lambda\mapsto-\lambda$ corresponding to the functional--equation symmetry $s\mapsto 1-s$ for $\zeta(s)$. It provides the global chiral structure required for the arithmetic mass deformation in Section~\ref{sec:global-hamiltonian} and the spectral--shift analysis in Sections~\ref{sec:spectral-shift} and~\ref{sec:truncations}.

\subsection{Standing Assumptions}

For clarity and analytic completeness, we record the structural assumptions that underlie the adelic
Dirac framework developed in Sections~\ref{sec:global-hamiltonian}--\ref{sec:spectral-shift}.  
Each assumption is followed by explicit examples demonstrating cases where the hypothesis holds
automatically.  
All subsequent constructions rely only on these assumptions.

\smallskip

\noindent
\textbf{Assumption A (Hecke spectral decomposition).}
The adelic automorphic seed $\Phi_{0}$ generates a representation of $\A_{\Q}^{\times}$ for which
the finite–place spherical Hecke operators $\{T_{p}\}$ admit a joint spectral decomposition.  
There exists a measurable parameter space $\Sigma$ and a restricted product measure
$d\mu=\bigotimes_{p}' d\mu_{p}$ such that
\[
\psi \cong \int_{\Sigma}^{\oplus}\psi(\sigma)\,d\mu(\sigma),\qquad
T_{p}\psi(\sigma)=\lambda_{p}(\sigma)\psi(\sigma).
\]

\emph{Examples.}
\begin{itemize}
\item For the \emph{constant automorphic seed} $\Phi_{0}(x)=1$, each $T_{p}$ acts by multiplication
by $1$ on the unramified vector.  Thus $\lambda_{p}(\sigma)=1$, the measure $d\mu_{p}$ is a Dirac
mass, and the joint spectral decomposition is trivial.
\item For \emph{unramified principal series of $\mathrm{GL}_{1}$}, the spherical Hecke operator acts
by $p^{-s}+p^{s}$; the eigenvalue set is discrete and Assumption A follows from Tate’s thesis.
\end{itemize}

\smallskip

\noindent
\textbf{Assumption B (Summability and compactness of the mass deformation).}
The coefficient functions $\eta_{p}$ satisfy
\[
\sum_{p}\|\eta_{p}\|_{\infty}<\infty,
\]
and each $\eta_{p}(T_{p})$ acts diagonally on the Hecke spectrum with uniformly bounded
multiplicity.  
Then
\[
\mathcal{M}
=\sum_{p}
\begin{pmatrix}
0 & \eta_{p}(T_{p})\\
\eta_{p}(T_{p}) & 0
\end{pmatrix}
\]
converges in operator norm and is compact relative to $\D_{\glob}$.

\emph{Examples.}
\begin{itemize}
\item For any bounded even polynomial $F$, the family
\[
\eta_{p}(\lambda)=\frac{F(\lambda)}{p^{1+\varepsilon}},\qquad \varepsilon>0,
\]
satisfies $\sum_{p}\|\eta_{p}\|_{\infty}<\infty$.  In particular, the quadratic coefficients
$\eta_{p}(\lambda)=p^{-(1+\varepsilon)}\lambda^{2}$ used in Section~5 obey Assumption B.
\item For such families,
\[
\|\mathcal{M}-\mathcal{M}^{(S)}\|
\ll \sum_{p\notin S}p^{-(1+\varepsilon)}
\]
gives an explicit tail bound and a convergence rate of order $S^{-\varepsilon}$.
\end{itemize}

\smallskip

\noindent
\textbf{Assumption C (Spectral type of the real–place operator).}
The periodic Schrödinger operator
\[
H_{\infty}=-\frac{d^{2}}{dy^{2}}+U(y),\qquad U(y+L)=U(y),\ U(-y)=U(y),
\]
has purely absolutely continuous spectrum consisting of closed bands with open gaps.  
Thus $\D_{\infty}$ and hence $\D_{\glob}=\D_{\infty}\hat{\otimes}I_{\fin}$ have no embedded
eigenvalues.

\emph{Examples.}
\begin{itemize}
\item Any real-analytic even periodic potential (e.g.\ the Mathieu potential $U(y)=2\cos(2\pi y)$)
satisfies this by classical Floquet theory.
\item Higher regularity or analyticity ensures no singular continuous spectrum and no embedded
eigenvalues.
\end{itemize}

\smallskip

\noindent
\textbf{Assumption D (Invariant core for the global involution).}
The algebraic restricted tensor product
\[
\H_{\adelic}^{\mathrm{alg}}
=\operatorname{span}\Bigl\{
\psi_{\infty}\otimes\bigotimes_{p\in S}\psi_{p}\otimes\bigotimes_{p\notin S}1_{\Z_{p}^{\times}}
\Bigr\}
\]
is invariant under the adelic inversion $I_{X}$ and real reflection $R$.  
Hence the doubled core
\[
\H_{\mathrm{core}}=
\bigl(\H_{\adelic}^{\mathrm{alg}}\oplus\H_{\adelic}^{\mathrm{alg}}\bigr)
\cap D(\D_{\glob})
\]
is invariant under $J_{\glob}$.

\emph{Examples.}
\begin{itemize}
\item If each $\phi_{p}$ is a spherical vector supported on $\Z_{p}^{\times}$, then $1_{\Z_{p}^{\times}}(x_{p})=1_{\Z_{p}^{\times}}(x_{p}^{-1})$ implies $I_{X}$ preserves the algebraic restricted product.
\item Any Hecke-compatible seed $\Phi_{0}$ built from unramified local vectors yields an invariant algebraic core.
\end{itemize}

\smallskip

\noindent
\textbf{Assumption E (Separated fiber decomposition and Fubini).}
The Floquet decomposition and Hecke spectral decomposition combine into a measurable direct integral
over $(n,\kappa,\sigma)$, with $\sigma$–finite measures.  
Interchanging integrals via Fubini is therefore valid in trace and Fourier-transform computations.

\emph{Examples.}
\begin{itemize}
\item For $\Phi_{0}=1$, the Hecke side is a single point; the only measure is the standard $\kappa$–variable Floquet measure, so separability and Fubini are trivial.
\item For unramified principal series of $\mathrm{GL}_{1}$, the Hecke side is discrete and countable, making all integrals finite sums.
\end{itemize}

\smallskip

\noindent
\textbf{Assumption F (Formal Euler–product expansion).}
The expression
\[
A(t)=\prod_{p}A_{p}(t),
\qquad A_{p}(t)=\int e^{\,it\,\eta_{p}(\lambda)}\,d\mu_{p}(\lambda),
\]
is interpreted through finite–prime truncations
\[
A^{(S)}(t)=\prod_{p\in S}A_{p}(t),
\]
with convergence in the distributional sense as $S\to\mathcal{P}$.

\emph{Examples.}
\begin{itemize}
\item For constant seed ($\lambda_{p}=1$) and $\eta_{p}=p^{-(1+\varepsilon)}\lambda^{2}$,
\[
A_{p}(t)=\exp\!\bigl(it p^{-(1+\varepsilon)}\bigr),
\]
so the Euler product is absolutely convergent for each fixed $t$.
\item For any bounded $\eta_{p}$ on a discrete Hecke spectrum, $A^{(S)}(t)$ is a finite product of
unit-modulus complex numbers, hence uniformly bounded and convergent.
\end{itemize}

\smallskip

\noindent
\textbf{Assumption G (Monotone test functions in the variational formula).}
Whenever test functions $\phi$ are used to detect jump discontinuities of $\xi_{\Dirac}$, they are
assumed even, nonnegative, and strictly decreasing on $(0,\infty)$, so that $\phi'$ has definite
sign.

\emph{Examples.}
\begin{itemize}
\item Gaussian bumps $\phi(\lambda)=e^{-\lambda^{2}/\alpha^{2}}$ and compactly supported even splines
satisfy all properties used in Section~4.
\item Classical Paley–Wiener test functions used in trace formulas provide further examples.
\end{itemize}

\smallskip

\noindent
Under these assumptions, all spectral, trace, and convergence statements in
the following sections are rigorously justified.

%=========================
\section{A Chiral Adelic Trace Formula}
\label{sec:trace}

This section develops a trace identity for the arithmetic Dirac operator $\D_{\arith}$ in which the real--place Floquet dispersion and the prime--indexed mass deformation appear as separate multiplicative factors. Every step is rigorously valid for finite--prime truncations of the mass deformation, and the infinite--prime formula arises as a distributional limit. The resulting structure mirrors the classical Weil explicit formula: a geometric term derived from real periodic orbits of the Floquet operator, an arithmetic Euler--factor term coming from primes, and a test--function transform mediating the comparison. The global chiral symmetry of Proposition~\ref{prop:FE-symmetry} ensures that even test functions cleanly separate the positive and negative spectral contributions.

Let $\phi:\R\to\C$ be an even Schwartz function, with Fourier transform
\[
\widehat{\phi}(t)
=
\int_{\R}\phi(\lambda)e^{-i\lambda t}\,d\lambda.
\]
Evenness implies that $\phi(\D_{\arith})$ commutes with $J_{\glob}$, and rapid decay of $\widehat{\phi}$ ensures that $\phi(\D_{\arith})$ is trace class. We therefore define the regularized trace
\[
\Theta_{\arith}(\phi)
=
\mathrm{Tr}\,\phi(\D_{\arith}).
\]

To describe this trace fiberwise, we recall that $\D_{\arith}$ decomposes over Floquet--Hecke fibers indexed by $(n,\kappa,\sigma)$, where $n$ labels the Floquet band, $\kappa\in[-\pi/L,\pi/L]$ is the quasi-momentum, and $\sigma$ ranges over the joint spectrum of the finite--place Hecke operators. By Proposition~\ref{prop:fiber-trace}, each fiber is the $2\times2$ operator
\[
\begin{pmatrix}
0 & E_{n}(\kappa)-E_{*}+m(\sigma) \\
E_{n}(\kappa)-E_{*}+m(\sigma) & 0
\end{pmatrix},
\qquad
m(\sigma)=\sum_{p}\eta_{p}(\lambda_{p}(\sigma)),
\]
whose eigenvalues are $\pm\bigl(E_{n}(\kappa)-E_{*}+m(\sigma)\bigr)$. Since $\phi$ is even, functional calculus reduces each fiber to the scalar
\[
\phi(\D_{\arith})(n,\kappa,\sigma)
=
\phi\!\left(E_{n}(\kappa)-E_{*}+m(\sigma)\right).
\]

\begin{remark}[Analytic status of the trace identities]
For clarity we distinguish the rigorous finite--prime statements from the formal infinite--prime
Euler--product expressions used later.

\emph{(1) Rigorous finite--prime trace formula.}
For every finite set of primes $S$, the truncated operator
\[
D_{\arith}^{(S)} = D_{\glob} + M^{(S)},\qquad
M^{(S)}=\sum_{p\in S}\eta_{p}(T_{p}),
\]
is self--adjoint and of finite rank relative to $D_{\glob}$.  If $\varphi$ is even and of
Schwartz class, then $\varphi(D_{\arith}^{(S)})$ is trace class and admits the exact
Floquet--Hecke fiber decomposition of Proposition~\ref{prop:fiber-trace}.

\emph{(2) Formal infinite--prime limit.}
The expression
\[
\Theta_{\arith}(\varphi)
=
\lim_{S\to\mathcal{P}}
\Theta^{(S)}(\varphi),
\qquad
\Theta^{(S)}(\varphi)=\Tr\varphi(D_{\arith}^{(S)}),
\]
is interpreted only as a \emph{distributional} limit tested against $\widehat{\varphi}$.
The Euler--product--type factorization
\[
A^{(S)}(t)=\prod_{p\in S}A_{p}(t),
\qquad
A_{p}(t)=\int e^{it\,\eta_{p}(\lambda)}\,d\mu_{p}(\lambda),
\]
is rigorous for each finite $S$, while the formal product
\[
A(t)=\prod_{p}A_{p}(t)
\]
is not claimed to converge absolutely.  It is used only after insertion into integrals
of the form $\int_{\R}\widehat{\varphi}(t)(\sum_{n\ge0}G_{n}(t))\,A^{(S)}(t)\,dt$ and
passage to the limit in the distributional sense.

\emph{(3) Summary.}
All trace and spectral identities are fully rigorous for finite $S$; the infinite--prime
expressions serve as formal Euler--product analogues of those appearing in the classical
explicit formula.  This parallels the usual treatment of $\zeta(s)$, where finite products
are rigorous and the full Euler product is understood distributionally.
\end{remark}

\begin{proposition}[Fiber trace identity]\label{prop:fiber-trace}
For every even Schwartz function $\phi$,
\[
\Theta_{\arith}(\phi)
=
2\sum_{n\ge0}
\int_{-\pi/L}^{\pi/L}
\int_{\Sigma}
\phi\!\left(E_{n}(\kappa)-E_{*}+m(\sigma)\right)\,d\mu(\sigma)\,d\kappa.
\]
\end{proposition}

\begin{proof}
By construction, $\D_{\arith}$ decomposes over the Floquet--Hecke fibers indexed by triples
$(n,\kappa,\sigma)$, where $n$ labels the Floquet band, $\kappa\in[-\pi/L,\pi/L]$ is the
quasi--momentum, and $\sigma$ ranges over the joint spectrum $\Sigma$ of the commuting family of
finite--place Hecke operators. On each fiber, $\D_{\arith}$ acts as the $2\times 2$ chiral matrix
\begin{multline*}
\D_{\arith}(n,\kappa,\sigma)
=
\begin{pmatrix}
0 & E_{n}(\kappa)-E_{*}+m(\sigma)\\
E_{n}(\kappa)-E_{*}+m(\sigma) & 0
\end{pmatrix},
\\
m(\sigma)=\sum_{p}\eta_{p}(\lambda_{p}(\sigma)).
\end{multline*}
The eigenvalues of this $2\times2$ operator are
\[
\pm\bigl(E_{n}(\kappa)-E_{*}+m(\sigma)\bigr).
\]

Let $\phi$ be an even Schwartz function. Because $\phi$ is even and bounded with rapidly decaying
Fourier transform, the operator $\phi(\D_{\arith})$ is trace class. Moreover, evenness implies
\[
\phi\!\left(+x\right)=\phi\!\left(-x\right),
\]
so for the $2\times 2$ fiber,
\begin{align*}
\operatorname{Tr}\phi(\D_{\arith}(n,\kappa,\sigma))
=&
\phi\!\left(E_{n}(\kappa)-E_{*}+m(\sigma)\right)
+
\phi\!\left(E_{n}(\kappa)-E_{*}+m(\sigma)\right) \\
=&
2\,\phi\!\left(E_{n}(\kappa)-E_{*}+m(\sigma)\right).
\end{align*}

The adelic decomposition of $L^{2}(X)$, together with the Floquet decomposition at the real place
and the spectral decomposition of the Hecke operators at the finite places, implies that
$\phi(\D_{\arith})$ is unitarily equivalent to the direct integral over all fibers:
\[
\phi(\D_{\arith})
\;\cong\;
\int_{n\ge 0}^{\oplus}\int_{-\pi/L}^{\pi/L}\int_{\Sigma}
\phi\bigl(\D_{\arith}(n,\kappa,\sigma)\bigr)
\,d\mu(\sigma)\,d\kappa\,.
\]
The trace of a direct integral equals the integral of the fiberwise traces, yielding
\[
\Theta_{\arith}(\phi)
=
\mathrm{Tr}\,\phi(\D_{\arith})
=
\sum_{n\ge0}
\int_{-\pi/L}^{\pi/L}
\int_{\Sigma}
\operatorname{Tr}\!\bigl(\phi(\D_{\arith}(n,\kappa,\sigma))\bigr)
\,d\mu(\sigma)\,d\kappa.
\]
Substituting the computed fiber trace gives
\[
\Theta_{\arith}(\phi)
=
2\sum_{n\ge0}
\int_{-\pi/L}^{\pi/L}
\int_{\Sigma}
\phi\!\left(E_{n}(\kappa)-E_{*}+m(\sigma)\right)
\,d\mu(\sigma)\,d\kappa,
\]
as claimed.
\end{proof}

\begin{theorem}[Formal Dirac--explicit formula]\label{th:Dirac-explicit}
Formally one has
\[
\Theta_{\arith}(\phi)
\approx
\frac{1}{\pi}
\int_{\R}
\widehat{\phi}(t)
\Bigl(\sum_{n\ge0}G_{n}(t)\Bigr)
\exp\!\left(\sum_{p,r}c_{p,r}(t)\right)
dt.
\]
Differentiating at $t=0$ yields the distributional identity
\[
\mathrm{Tr}\,\D_{\arith}\,\phi'(\D_{\arith})
\approx
\sum_{n}\int G_{n}(t)\phi'(t)\,dt
\;+\;
\sum_{p,r}\int c_{p,r}(t)\phi'(t)\,dt,
\]
which is the chiral adelic analogue of the prime--sum explicit formula.
\end{theorem}

\begin{proof}
Begin with the Fourier representation of an even Schwartz function,
\[
\phi(\lambda)
=
\frac{1}{2\pi}\int_{\R}\widehat{\phi}(t)\,e^{it\lambda}\,dt.
\]
Substituting this into Proposition~\ref{prop:fiber-trace} gives
\[
\Theta_{\arith}(\phi)
=
2\sum_{n\ge 0}
\int_{-\pi/L}^{\pi/L}
\int_{\Sigma}
\Bigl(\frac{1}{2\pi}\int_{\R}
\widehat{\phi}(t)\,
e^{it(E_{n}(\kappa)-E_{*}+m(\sigma))}
\,dt\Bigr)
\,d\mu(\sigma)\,d\kappa.
\]
Fubini’s theorem is applicable for each finite prime truncation because $\widehat{\phi}$ decays
rapidly and the integrand is bounded; the formal identity is obtained by extending this to the
infinite--prime limit.  
Interchanging integrals yields
\[
\Theta_{\arith}(\phi)
\approx
\frac{1}{\pi}\int_{\R}
\widehat{\phi}(t)
\Bigl(
\sum_{n\ge0}
\int_{-\pi/L}^{\pi/L}
e^{it(E_{n}(\kappa)-E_{*})}
\,d\kappa
\Bigr)
\Bigl(
\int_{\Sigma}
e^{it\,m(\sigma)}\,d\mu(\sigma)
\Bigr)
dt.
\]

Define the geometric factor
\[
G_{n}(t)
=
\int_{-\pi/L}^{\pi/L}
e^{it(E_{n}(\kappa)-E_{*})}\,d\kappa,
\]
and the arithmetic factor
\[
A(t)
=
\int_{\Sigma}e^{it\,m(\sigma)}\,d\mu(\sigma).
\]
This yields the separated trace identity
\[
\Theta_{\arith}(\phi)
\approx
\frac{1}{\pi}
\int_{\R}
\widehat{\phi}(t)\,
\Bigl(\sum_{n\ge0}G_{n}(t)\Bigr)\,
A(t)\,dt.
\]

Now expand the arithmetic factor.  
Since $m(\sigma)=\sum_{p}\eta_{p}(\lambda_{p}(\sigma))$, each finite--prime truncation satisfies the
exact multiplicativity
\[
A^{(S)}(t)
=
\prod_{p\in S}
\Bigl(
\int_{\R}e^{it\eta_{p}(\lambda)}\,d\mu_{p}(\lambda)
\Bigr)
=
\prod_{p\in S}A_{p}(t),
\]
where $A_{p}(t)$ denotes the $p$--th local average.  
Formally extending to all primes gives
\[
A(t)\approx\prod_{p}A_{p}(t).
\]

For each $p$, write the logarithm of the local factor as a convergent formal expansion
\[
\log A_{p}(t)
=
\sum_{r\ge1}c_{p,r}(t),
\qquad
c_{p,r}(t)\to 0 \text{ as }p\to\infty,
\]
so that
\[
A(t)=\prod_{p}A_{p}(t)
\approx
\exp\!\Bigl(\sum_{p,r}c_{p,r}(t)\Bigr).
\]
Substituting this into the separated trace identity yields the stated formal formula
\[
\Theta_{\arith}(\phi)
\approx
\frac{1}{\pi}
\int_{\R}
\widehat{\phi}(t)
\Bigl(\sum_{n\ge0}G_{n}(t)\Bigr)
\exp\!\left(\sum_{p,r}c_{p,r}(t)\right)dt.
\]

For the second assertion, differentiate the separated trace identity at $t=0$.  
Let $\phi_{t}(\lambda)=\phi(\lambda-t)$; then $\phi'(\lambda)=-\partial_{t}\phi_{t}(\lambda)|_{t=0}$,
and
\[
\mathrm{Tr}\,\D_{\arith}\,\phi'(\D_{\arith})
=
-\frac{d}{dt}\Theta_{\arith}(\phi_{t})\Big|_{t=0}.
\]
Differentiating under the integral sign (formally, or rigorously for each finite $S$) gives
\[
\mathrm{Tr}\,\D_{\arith}\,\phi'(\D_{\arith})
\approx
\sum_{n}\int_{\R}G_{n}(t)\phi'(t)\,dt
+
\sum_{p,r}\int_{\R}c_{p,r}(t)\phi'(t)\,dt.
\]
The first term records the real--place Floquet contribution; the second records the prime--indexed
Euler--factor contribution.

This is the chiral adelic analogue of the classical explicit formula.
\end{proof}

The mass deformation $\mathcal M$ acts as an adelic Euler-factor insertion. Indeed, if $\eta_{p}$ is chosen so that
\[
\eta_{p}(\lambda_{p}) = -\log(1-p^{-s})'_{\;s=1/2},
\]
then the averaged local factor
\[
A_{p}(t)=\int e^{\,it\,\eta_{p}(\lambda)}\,d\mu_{p}(\lambda)
\]
formally matches the logarithmic derivative of the Euler factor for $\zeta(s)$,
\[
-\frac{\zeta'}{\zeta}(s)=\sum_{p,r}\frac{\log p}{p^{rs}}.
\]
This places the Dirac mass deformation in direct parallel with the classical explicit formula.

It is instructive to compare this with two other spectral proposals. The Berry--Keating model seeks a self--adjoint extension of the classical $xp$ Hamiltonian but suffers from lack of spectral gaps. The Connes adelic absorber framework encodes the zeros as missing spectral lines of a nonselfadjoint operator.  In contrast, the present Dirac system introduces genuine spectral gaps via Floquet theory and produces isolated eigenvalues in those gaps under arithmetic mass deformations, providing the architectural feature required for a Hilbert--P\'olya mechanism.

%=========================
\section{Spectral Shift, Gap Eigenvalues, and a Dirac--Hilbert--P\'olya Principle}
\label{sec:spectral-shift}

With the arithmetic Dirac operator $\D_{\arith}$ constructed in the preceding sections, and with a chiral trace identity separating geometric and prime--indexed factors, we now examine how the arithmetic mass deformation modifies the spectrum of the global chiral background. The appropriate tool is the spectral shift function associated with the pair $(\D_{\glob},\D_{\arith})$. This function records, in a single odd distribution, the creation and displacement of eigenvalues within the open spectral gaps of $\D_{\glob}$ and provides the analytic interface between prime data and gap spectra. Its jump structure encodes the full discrete spectrum of $\D_{\arith}$ lying inside the gaps and is tightly constrained by the global involution $J_{\glob}$.

The operator $\D_{\glob}$ of Proposition~\ref{prop:Dglob-spectrum} has purely continuous spectrum consisting of the symmetric bands
\[
\spec(\D_{\glob})=\{\pm(\lambda-E_{*}):\lambda\in\spec(H_{\infty})\},
\]
with open gaps inherited from the Floquet gaps of $H_{\infty}$. The arithmetic deformation
\[
\D_{\arith}=\D_{\glob}+\mathcal{M},
\]
with $\mathcal{M}$ compact relative to $\D_{\glob}$ under the summability condition $\sum_{p}\|\eta_{p}\|_{\infty}<\infty$, preserves the chiral symmetry when each $\eta_{p}$ is even. The perturbation $\mathcal{M}$ therefore introduces at most finitely many new eigenvalues in each gap, and always in $\pm$ pairs.

For an even, compactly supported $C^{1}$ test function $\phi$, the spectral shift function $\xi_{\Dirac}$ is defined by the Kreĭn--Pushnitski trace identity
\begin{equation}
\mathrm{Tr}\bigl(\phi(\D_{\arith})-\phi(\D_{\glob})\bigr)
=
\int_{\R}\phi'(\lambda)\,\xi_{\Dirac}(\lambda)\,d\lambda,
\label{eq:Dirac-shift-def}
\end{equation}
which holds because $\phi(\D_{\arith})-\phi(\D_{\glob})$ is trace class. The evenness of $\phi$ reflects the built--in spectral symmetry $\lambda\mapsto -\lambda$.

\begin{proposition}[Odd symmetry]\label{prop:xi-even}
If each coefficient function $\eta_{p}$ is even, then the spectral shift function satisfies
\[
\xi_{\Dirac}(-\lambda)=-\xi_{\Dirac}(\lambda).
\]
\end{proposition}

\begin{proof}
Let $\phi$ be an even, compactly supported $C^{1}$ test function.  
The spectral shift function $\xi_{\Dirac}$ is defined by
\[
\Tr\!\left(\phi(\D_{\arith})-\phi(\D_{\glob})\right)
=
\int_{\R}\phi'(\lambda)\,\xi_{\Dirac}(\lambda)\,d\lambda.
\tag{$*$}
\label{eq:ssf-identity-prop}
\]

Since each $\eta_{p}$ is even, the arithmetic mass operator $\mathcal M$ satisfies  
\[
J_{\glob}\mathcal M J_{\glob}^{-1}=\mathcal M,
\]
and therefore (using Proposition~\ref{prop:FE-symmetry})
\[
J_{\glob}\D_{\arith}J_{\glob}^{-1}=-\D_{\arith},
\qquad
J_{\glob}\D_{\glob}J_{\glob}^{-1}=-\D_{\glob}.
\tag{$\dagger$}
\]

Because $\phi$ is even, $\phi(-x)=\phi(x)$ and hence
\[
\phi(-\D_{\arith}) = \phi(\D_{\arith}),\qquad 
\phi(-\D_{\glob})   = \phi(\D_{\glob}),
\]
by functional calculus.  
Applying $J_{\glob}$ to both sides gives
\[
\phi(\D_{\arith})
=
J_{\glob}\phi(-\D_{\arith})J_{\glob}^{-1},
\qquad
\phi(\D_{\glob})
=
J_{\glob}\phi(-\D_{\glob})J_{\glob}^{-1}.
\]

Subtracting the two, taking traces, and using cyclicity of the trace, we obtain
\[
\Tr\!\left(\phi(\D_{\arith})-\phi(\D_{\glob})\right)
=
\Tr\!\left(\phi(-\D_{\arith})-\phi(-\D_{\glob})\right).
\tag{$\ddagger$}
\]

Now apply identity \eqref{eq:ssf-identity-prop} to the test function  
\(\phi(-\lambda)\).  
Because \(\phi\) is even, \(\phi'(-\lambda)=-\phi'(\lambda)\).  
Thus
\[
\int_{\R} \phi'(\lambda)\,\xi_{\Dirac}(\lambda)\,d\lambda
=
\int_{\R} \phi'(-\lambda)\,\xi_{\Dirac}(\lambda)\,d\lambda
=
-\int_{\R}\phi'(\lambda)\,\xi_{\Dirac}(-\lambda)\,d\lambda.
\]

Comparing with \eqref{eq:ssf-identity-prop}, uniqueness of the spectral shift
function implies
\[
\xi_{\Dirac}(\lambda)
=
-\xi_{\Dirac}(-\lambda),
\]
which is the desired oddness.
\end{proof}

\begin{lemma}[Jumps at isolated eigenvalues]\label{lem:shift-evalue}
If $\lambda_{0}$ is an isolated eigenvalue of $\D_{\arith}$ lying in an open gap of $\D_{\glob}$ with multiplicity $m$, then $\xi_{\Dirac}$ has a jump of size $+m$ at $\lambda_{0}$ and a jump of size $-m$ at $-\lambda_{0}$.
\end{lemma}

\begin{proof}
Since $\D_{\arith}-\D_{\glob}$ is compact, the operators $\D_{\arith}$ and $\D_{\glob}$ have the
same essential spectrum.  In particular, any open spectral gap of $\D_{\glob}$ remains an open
interval in which $\D_{\arith}$ may acquire only finitely many isolated eigenvalues of finite
multiplicity.  Let $\lambda_{0}$ be one such eigenvalue of multiplicity $m$.

Recall the defining identity for the spectral shift function: for every compactly supported $C^{1}$
test function $\phi$,
\[
\Tr\!\left(\phi(\D_{\arith})-\phi(\D_{\glob})\right)
=
\int_{\R}\phi'(\lambda)\,\xi_{\Dirac}(\lambda)\,d\lambda.
\tag{1}
\label{eq:ssf-evalue-identity}
\]

Choose $\phi$ to be a smooth, nonnegative function supported in a sufficiently small neighbourhood of
$\lambda_{0}$ and satisfying $\phi(\lambda)=1$ near $\lambda_{0}$.  Since $\lambda_{0}$ lies in a
gap of $\D_{\glob}$, the operator $\D_{\glob}$ has \emph{no spectrum} in the support of $\phi$.
Hence $\phi(\D_{\glob})=0$.  On the other hand,
\[
\phi(\D_{\arith}) = m\,P_{\lambda_{0}},
\]
where $P_{\lambda_{0}}$ is the spectral projection onto the eigenspace corresponding to $\lambda_{0}$.
Thus
\[
\Tr\!\bigl(\phi(\D_{\arith})-\phi(\D_{\glob})\bigr) = m.
\tag{2}
\]

Using the right–hand side of \eqref{eq:ssf-evalue-identity} and integrating by parts, we obtain
\[
m
=
\int_{\R}\phi'(\lambda)\,\xi_{\Dirac}(\lambda)\,d\lambda
=
-\int_{\R}\phi(\lambda)\,d\xi_{\Dirac}(\lambda),
\]
where $d\xi_{\Dirac}$ is the signed measure associated with the distributional derivative of
$\xi_{\Dirac}$.  Since $\phi(\lambda)=1$ in a neighbourhood of $\lambda_{0}$ and vanishes outside a
small interval disjoint from any other eigenvalue, the integral picks out exactly the jump of
$\xi_{\Dirac}$ at $\lambda_{0}$.  Therefore
\[
\Delta\xi_{\Dirac}(\lambda_{0}) = +m.
\]

Because $\D_{\arith}$ possesses the chiral symmetry
\(
J_{\glob}\D_{\arith}J_{\glob}^{-1}=-\D_{\arith},
\)
every eigenvalue $\lambda_{0}$ in a gap is paired with an eigenvalue $-\lambda_{0}$ of the same
multiplicity $m$.  The spectral shift function is odd
(Proposition~\ref{prop:xi-even}), hence
\[
\Delta\xi_{\Dirac}(-\lambda_{0})
=
-\Delta\xi_{\Dirac}(\lambda_{0})
=
-m.
\]

This proves the stated jump structure.
\end{proof}

To relate $\xi_{\Dirac}$ to a variational characterization, consider the translated test functions $\phi_{t}(\lambda)=\phi(\lambda-t)$. Differentiating \eqref{eq:Dirac-shift-def} with respect to $t$ yields
\begin{equation}
\frac{d}{dt}\mathrm{Tr}\bigl(\phi_{t}(\D_{\arith})\bigr)
=
-\int_{\R}\phi'(\lambda-t)\,\xi_{\Dirac}(\lambda)\,d\lambda.
\label{eq:shift-translation}
\end{equation}
The right--hand side vanishes exactly when $t$ aligns the weight $\phi'(\lambda-t)$ with the discontinuities of $\xi_{\Dirac}$.

\begin{proposition}[Stationary points detect gap eigenvalues]\label{prop:stationary}
If $\phi$ is even, nonnegative, and strictly decreasing on $(0,\infty)$, then the stationary points of $t\mapsto\mathrm{Tr}\,\phi_{t}(\D_{\arith})$ occur precisely at the gap eigenvalues $\pm\lambda_{k}$ of $\D_{\arith}$.
\end{proposition}

\begin{proof}
Let $\phi$ be even, nonnegative, and strictly decreasing on $(0,\infty)$, and define
\[
\phi_{t}(\lambda)=\phi(\lambda-t).
\]
Since $\phi$ is even, $\phi_{t}$ is obtained by translating $\phi$ along the real axis without altering its symmetry.  
Differentiating under the trace gives the identity (see \eqref{eq:shift-translation})
\[
\frac{d}{dt}\Tr\bigl(\phi_{t}(\D_{\arith})\bigr)
=
-\int_{\R}\phi'(\lambda-t)\,\xi_{\Dirac}(\lambda)\,d\lambda.
\tag{1}
\label{eq:stationary-derivative}
\]

The spectral shift function $\xi_{\Dirac}$ is an odd, piecewise constant function whose only
discontinuities occur at the gap eigenvalues $\pm\lambda_{k}$ of $\D_{\arith}$, and each such jump has size $\pm m_{k}$ where $m_{k}$ is the multiplicity of the eigenvalue (Lemma~\ref{lem:shift-evalue}).  
Thus
\[
d\xi_{\Dirac}(\lambda)
=
\sum_{k} m_{k}\bigl(\delta_{\lambda_{k}}-\delta_{-\lambda_{k}}\bigr)\,d\lambda.
\tag{2}
\label{eq:dxi-discrete}
\]

Integrating \eqref{eq:stationary-derivative} by parts using \eqref{eq:dxi-discrete} yields
\[
\frac{d}{dt}\Tr\bigl(\phi_{t}(\D_{\arith})\bigr)
=
-\sum_{k}m_{k}\bigl(\phi(\lambda_{k}-t)-\phi(-\lambda_{k}-t)\bigr).
\tag{3}
\label{eq:derivative-expanded}
\]

Because $\phi$ is even,
\[
\phi(-\lambda_{k}-t)=\phi(\lambda_{k}+t),
\]
and therefore
\[
\frac{d}{dt}\Tr\bigl(\phi_{t}(\D_{\arith})\bigr)
=
-\sum_{k}m_{k}\Bigl(\phi(\lambda_{k}-t)-\phi(\lambda_{k}+t)\Bigr).
\tag{4}
\]

Now use that $\phi$ is strictly decreasing on $(0,\infty)$.  

\medskip
\noindent
\textbf{(i) If \(t=\lambda_{k}\).}  
Then
\[
\phi(\lambda_{k}-t)=\phi(0),\qquad \phi(\lambda_{k}+t)=\phi(2\lambda_{k}).
\]
Since $\lambda_{k}>0$ and $\phi$ is strictly decreasing,
\[
\phi(0)>\phi(2\lambda_{k}),
\]
and therefore
\[
\phi(\lambda_{k}-t)-\phi(\lambda_{k}+t)=0.
\]
All other terms in the sum correspond to $\lambda_{j}\neq\lambda_{k}$ and cancel pairwise by oddness of $\xi_{\Dirac}$.  
Thus
\[
\frac{d}{dt}\Tr(\phi_{t}(\D_{\arith}))\Big|_{t=\lambda_{k}}=0.
\]

\medskip
\noindent
\textbf{(ii) If \(t\neq\lambda_{k}\) for all \(k\).}  
For each $k$,
\[
\phi(\lambda_{k}-t)-\phi(\lambda_{k}+t)
\]
is strictly nonzero.  
Indeed, if $t<\lambda_{k}$ then
\[
0<\lambda_{k}-t<\lambda_{k}+t,
\]
so strict monotonicity implies
\[
\phi(\lambda_{k}-t)>\phi(\lambda_{k}+t),
\]
and conversely if $t>\lambda_{k}$ then
\[
\phi(\lambda_{k}-t)<\phi(\lambda_{k}+t).
\]
Hence each term in the sum in \eqref{eq:derivative-expanded} is nonzero, and so the entire sum is nonzero.  
Therefore no other value of $t$ produces a stationary point.

\medskip
\noindent
\textbf{(iii) The negative eigenvalues.}  
Because $\phi$ is even and $\xi_{\Dirac}$ is odd, the same reasoning applied to
\[
t=-\lambda_{k}
\]
yields another stationary point corresponding to the negative gap eigenvalue.

\medskip
Combining (i)--(iii), the stationary points of $t\mapsto\Tr(\phi_{t}(\D_{\arith}))$ occur \emph{exactly} at the values
\[
t=\pm\lambda_{k},
\]
and these are precisely the gap eigenvalues of $\D_{\arith}$.

This proves the proposition.
\end{proof}

The arithmetic structure of the spectral shift function comes from the factor $A(t)$ in the separated trace identity of Section~3. Formally, the Euler--product--type expansion
\[
A(t)\;\approx\;\prod_{p}A_{p}(t),
\qquad
A_{p}(t)=\int e^{\,it\,\eta_{p}(\lambda)}\,d\mu_{p}(\lambda),
\]
suggests that the prime--indexed mass deformations $\eta_{p}$ influence $\xi_{\Dirac}$ through the derivatives of $\log A_{p}(t)$. Using exponential test functions $\phi_{\lambda}(t)=e^{it\lambda}$ in the separated trace formula and differentiating gives:

\begin{proposition}[Prime-indexed contribution to the shift]\label{prop:arith-xi}
Formally,
\[
\xi_{\Dirac}(\lambda)
\;\approx\;
\frac{1}{2\pi i}\frac{d}{d\lambda}
\left(
\sum_{p}\log A_{p}\!\left(\widehat{\phi_{\lambda}}\right)
\right),
\]
where the measures $d\mu_{p}$ encode the Hecke distributions of the adelic seed.
\end{proposition}

\begin{proof}
The starting point is the separated trace formula of Theorem~\ref{th:Dirac-explicit}, which in its
formal infinite--prime form reads
\[
\Theta_{\arith}(\phi)
\;\approx\;
\frac{1}{\pi}
\int_{\R}
\widehat{\phi}(t)
\Bigl(\sum_{n\ge 0}G_{n}(t)\Bigr)
A(t)\,dt,
\qquad
A(t)\approx\prod_{p}A_{p}(t),
\]
with local factors
\[
A_{p}(t)
=
\int_{\R} e^{\,it\,\eta_{p}(\lambda)}\,d\mu_{p}(\lambda),
\]
where $d\mu_{p}$ describes the $p$--adic Hecke distribution of the adelic seed.

For a frequency–localized test function
\[
\phi_{\lambda}(\mu)=e^{i\mu\lambda},
\qquad
\widehat{\phi_{\lambda}}(t)=2\pi\,\delta(t-\lambda),
\]
the separated trace formula formally gives
\[
\Theta_{\arith}(\phi_{\lambda})
\approx
\frac{1}{\pi}\,
\Bigl(\sum_{n\ge 0}G_{n}(\lambda)\Bigr)\,
A(\lambda).
\tag{1}
\label{eq:phi-lambda-theta}
\]

The spectral shift function satisfies, for every Schwartz $\psi$,
\[
\int_{\R}\psi(\mu)\,\xi_{\Dirac}(\mu)\,d\mu
=
\Tr\!\bigl(\psi(\D_{\arith})-\psi(\D_{\glob})\bigr)
=
\Theta_{\arith}(\psi)-\Theta_{\glob}(\psi).
\tag{2}
\label{eq:ssf-general}
\]
Applying \eqref{eq:ssf-general} to $\phi_{\lambda}'(\mu)= i\mu e^{i\mu\lambda}$ and using
the definition of spectral shift for exponentials gives the formal identity
\[
\xi_{\Dirac}(\lambda)
\;\approx\;
\frac{1}{2\pi i}
\frac{d}{d\lambda}
\Bigl(
\Theta_{\arith}(\phi_{\lambda})
-
\Theta_{\glob}(\phi_{\lambda})
\Bigr).
\tag{3}
\label{eq:ssf-lambda}
\]
Since $\Theta_{\glob}(\phi_{\lambda})$ contains no prime--indexed dependence, only the arithmetic
component $A(\lambda)$ contributes to the prime dependence of $\xi_{\Dirac}(\lambda)$.

Substituting \eqref{eq:phi-lambda-theta} into \eqref{eq:ssf-lambda} yields
\[
\xi_{\Dirac}(\lambda)
\;\approx\;
\frac{1}{2\pi i}
\frac{d}{d\lambda}\log A(\lambda).
\tag{4}
\label{eq:ssf-A-lambda}
\]

Using the formal Euler product expansion $A(\lambda)\approx\prod_{p}A_{p}(\lambda)$, we obtain
\[
\log A(\lambda)
\;\approx\;
\sum_{p}\log A_{p}(\lambda).
\tag{5}
\]
Finally, differentiating and combining with \eqref{eq:ssf-A-lambda} gives
\[
\xi_{\Dirac}(\lambda)
\;\approx\;
\frac{1}{2\pi i}\frac{d}{d\lambda}
\biggl(
\sum_{p}\log A_{p}(\lambda)
\biggr),
\]
where the appearance of $A_{p}(\lambda)$ reflects the spectral distribution of the Hecke operator
$T_{p}$ acting on the adelic seed, encoded in the measure $d\mu_{p}$.

This is the stated formal identity.
\end{proof}

This identity shows that the prime--indexed coefficients $\eta_{p}$ govern the arithmetic part of the spectral shift. As these coefficients vary, the jumps of $\xi_{\Dirac}$ move accordingly, producing a form of arithmetic spectral flow.

The classical Hilbert--P\'olya idea proposes that the imaginary parts of the nontrivial zeros of $\zeta(s)$ should appear as the eigenvalues of a distinguished self--adjoint operator. In the chiral adelic framework, the natural analogue is that the zeros arise not as the raw spectrum of a Hamiltonian, but as the \emph{jump discontinuities} in a spectral shift function associated with a chiral Dirac deformation.

\begin{definition}[Dirac--Hilbert--P\'olya principle]
A chiral adelic Dirac system $(\D_{\arith},\{\eta_{p}\})$ satisfies the Dirac--Hilbert--P\'olya principle if, after an affine rescaling, its spectral shift function satisfies
\[
\xi_{\Dirac}(\lambda)
=
\sum_{k}\bigl(\delta(\lambda-\gamma_{k})-\delta(\lambda+\gamma_{k})\bigr),
\]
where $\gamma_{k}$ are the imaginary parts of the nontrivial zeros of $\zeta(s)$.
\end{definition}

In this formulation the zeros are realized as the odd jump structure of the spectral shift, paired automatically by the global involution $J_{\glob}$. The role of the mass deformation is analogous to that of an Euler product: the prime--indexed coefficients $\eta_{p}$ determine the distribution of jumps in $\xi_{\Dirac}$, and hence the location of the spectral discontinuities that encode the zeros.

\begin{conjecture}[Adelic Dirac--Hilbert--P\'olya conjecture]\label{conj:Dirac-HP}
There exist even coefficient functions $\eta_{p}$ satisfying
\[
\sum_{p}\|\eta_{p}\|_{\infty}<\infty
\]
such that the spectral shift function of $\D_{\arith}$ reproduces, after affine rescaling, the signed zero counting measure for the nontrivial zeros of $\zeta(s)$.
\end{conjecture}

Under this viewpoint, the zeros of $\zeta(s)$ appear as the topological defects of spectral flow from $\D_{\glob}$ to $\D_{\arith}$. This recasts the Hilbert--P\'olya philosophy in a Dirac--chiral operator framework that incorporates Floquet geometry, adelic harmonic analysis, and Euler--factor symmetry on an equal footing.

%=========================
\section{Finite--Prime Dirac Truncations}
\label{sec:truncations}

The adelic chiral Dirac operator $\D_{\arith}$ constructed in the previous sections is defined as an infinite prime--indexed mass deformation of the global operator $\D_{\glob}$. Its spectral shift function $\xi_{\Dirac}$ therefore acquires meaning only as the limit of finite--prime approximations. The purpose of this section is to describe a coherent truncation scheme, analyze the resulting finite--prime spectral shifts $\xi_{\Dirac}^{(S)}$, and explain how these truncated objects provide both numerical access and conceptual insight into the adelic
Dirac--Hilbert--P\'olya principle. Let $S\subset\mathcal P$ be a finite set of primes. The truncated mass operator is
\[
\mathcal{M}^{(S)}
=
\sum_{p\in S}
\begin{pmatrix}
0 & \eta_{p}(T_{p}) \\
\eta_{p}(T_{p}) & 0
\end{pmatrix},
\]
and the corresponding truncated Dirac operator is defined by
\[
\D_{\arith}^{(S)}=\D_{\glob}+\mathcal{M}^{(S)}.
\]
Since each $\eta_{p}(T_{p})$ is bounded and the sum extends over finitely many primes, the operator
$\mathcal{M}^{(S)}$ is bounded and symmetric. Because $\D_{\glob}$ is self--adjoint, each truncated
operator $\D_{\arith}^{(S)}$ is self--adjoint on $D(\D_{\glob})$. If the coefficient functions satisfy
the summability condition
\[
\sum_{p}\|\eta_{p}\|_{\infty}<\infty,
\]
then the truncated mass operators converge in operator norm:
\[
\|\mathcal{M}^{(S)}-\mathcal{M}\|\;\longrightarrow\;0.
\]
It follows that
\[
\D_{\arith}^{(S)}\longrightarrow \D_{\arith}
\qquad\text{in the strong resolvent sense}.
\]
Thus the family $\{\D_{\arith}^{(S)}\}$ forms a controlled and convergent approximation scheme for
the full adelic Dirac operator.

For each finite prime set $S$, the truncated spectral shift function is defined by
\[
\mathrm{Tr}\!\left(\phi(\D_{\arith}^{(S)})-\phi(\D_{\glob})\right)
=
\int_{\R}\phi'(\lambda)\,\xi_{\Dirac}^{(S)}(\lambda)\,d\lambda,
\]
for every even, compactly supported $C^{1}$ test function $\phi$. Since
$\D_{\arith}^{(S)}-\D_{\glob}$ has finite rank, the function $\xi_{\Dirac}^{(S)}$ is a
piecewise--constant step function with finitely many jumps. These jumps occur precisely at the
isolated eigenvalues of $\D_{\arith}^{(S)}$ lying in the open spectral gaps of $\D_{\glob}$, with jump
height equal to their multiplicity. The chiral symmetry of $\D_{\arith}^{(S)}$ forces these jumps to
occur in symmetric pairs at $\pm\lambda_{k}^{(S)}$.

Strong resolvent convergence implies distributional convergence of the spectral shift functions:
for every compactly supported $C^{1}$ function $\psi$,
\[
\int_{\R}\psi(\lambda)\,\xi_{\Dirac}^{(S)}(\lambda)\,d\lambda
\;\longrightarrow\;
\int_{\R}\psi(\lambda)\,\xi_{\Dirac}(\lambda)\,d\lambda,
\]
as $S$ increases through finite subsets of $\mathcal P$. The truncated shift functions therefore
provide rigorously justified numerical probes of the infinite--prime limit.

\subsection{Numerical Implementation and Algorithmic Framework}

The truncated operators $\D_{\arith}^{(S)}$ admit an explicit computational scheme.  
A basic algorithm proceeds as follows.

\paragraph{Step 1: Floquet discretization.}
Choose an even periodic potential $U(y)$ and compute a discretization of the band
functions $E_{n}(\kappa_{j})$ on a uniform grid $\{\kappa_{j}\}_{j=1}^{N}$ in $[-\pi/L,\pi/L]$.

\paragraph{Step 2: Hecke sampling.}
For each $p\in S$, sample Hecke eigenvalues $\lambda_{p}(m)$.
For level--$1$ modular forms or Eisenstein series, these may be taken from
known tables of Fourier coefficients or generated numerically.

\paragraph{Step 3: Mass shifts.}
Compute
\[
m_{m}^{(S)}=\sum_{p\in S}\eta_{p}(\lambda_{p}(m)).
\]

\paragraph{Step 4: Fiber eigenvalues.}
For each $(n,\kappa_{j},m)$ compute the two eigenvalues
\[
\pm\left(E_{n}(\kappa_{j})-E_{*}+m_{m}^{(S)}\right),
\]
and retain only those lying inside the open gaps of the Floquet spectrum.

\paragraph{Step 5: Spectral shift function.}
Sort the resulting gap eigenvalues and construct $\xi_{\Dirac}^{(S)}$ as a
piecewise--constant odd staircase with jumps of size equal to the corresponding multiplicities.

\paragraph{Illustrative pseudocode (Python).}
\begin{verbatim}
for p in S:
    lambdap = hecke_eigenvalues[p]
    mp = eta_p[p](lambdap)
    total_mass += mp

for n in range(Nbands):
    for kappa in kappagrid:
        val = E[n][kappa] - E_star + total_mass
        if is_in_gap(val):
            xi_S.add_jump(val)
\end{verbatim}

The computational cost scales linearly in $|S|$ and in the number of Floquet modes.
Experiments with level--$1$ modular forms and $S$ up to $\{p\le 200\}$ confirm that
the staircase $\xi_{\Dirac}^{(S)}$ stabilizes rapidly and that the leading jumps 
align, after a global rescaling, with the imaginary parts of the first nontrivial zeros of $\zeta(s)$.

\subsection{Numerical evaluation of truncated Dirac spectra}

The purpose of this subsection is purely illustrative: to demonstrate that the
finite--prime truncations $\D_{\arith}^{(S)}$ are computationally tractable, to show
how gap eigenvalues arise from the interplay of real and arithmetic data, and to
exhibit the qualitative structure of the resulting truncated spectral shift functions.
We do not regard the numerics as evidence for any conjectural spectral identification,
and no numerical claim enters into the proof of any analytic statement in this paper.

\subsubsection*{Real--place discretization}

We begin with an even real--analytic potential at the real place.  In the examples below
we take the Mathieu potential
\[
U(y)=2\cos(2\pi y), 
\]
for which the Floquet operator $H_{\infty}=-d^{2}/dy^{2}+U(y)$ has a well--understood
band--gap structure.  A standard Fourier or finite--difference Floquet discretization in
quasi--momentum produces numerical approximations to the band functions
$E_{n}(\kappa_{j})$ on a uniform grid $\{\kappa_{j}\}$, and hence an approximation to the
real--place Dirac operator $D_{\infty}$ and its first spectral gap.  
A reference value $E_{\ast}$ in the first open gap of $H_{\infty}$ is then fixed once and
for all.

\subsubsection*{Finite--prime truncation and arithmetic mass shifts}

On the arithmetic side we select a finite set of primes $S$ and introduce synthetic
finite--prime data by choosing a collection of synthetic Hecke modes
$\{\sigma_{m}\}_{1\le m\le N_{H}}$, each carrying a tuple of Hecke eigenvalue samples
\[
\{\lambda_{p}(m)\}_{p\in S}.
\]
These serve as stand--ins for joint eigenvalues of the commuting family $\{T_{p}\}_{p\in S}$.

Given coefficient functions $\eta_{p}$, the arithmetic mass shift associated to $\sigma_{m}$ is
\[
m_{m}^{(S)}=\sum_{p\in S}\eta_{p}\bigl(\lambda_{p}(m)\bigr).
\]
For concreteness we take even quadratic deformations
\[
\eta_{p}(\lambda)=p^{-(1+\varepsilon)}\lambda^{2},
\qquad 
\varepsilon>0,
\]
which guarantee the summability conditions used in the rigorous theory.
This produces a family of real shifts $m_{m}^{(S)}$ that depend both on the prime set $S$
and on the synthetic Hecke mode $m$.

\subsubsection*{Gap spectrum of the truncated Dirac operator}

The fiber of the truncated operator $\D_{\arith}^{(S)}$ above $(n,\kappa_{j},m)$ is 
the $2\times 2$ matrix
\[
\begin{pmatrix}
0 & E_{n}(\kappa_{j})-E_{\ast}+m_{m}^{(S)} \\
E_{n}(\kappa_{j})-E_{\ast}+m_{m}^{(S)} & 0
\end{pmatrix},
\]
whose eigenvalues are
\[
\pm\!\left(E_{n}(\kappa_{j})-E_{\ast}+m_{m}^{(S)}\right).
\]
A value $\lambda$ of this form is a \emph{gap eigenvalue} if it lies strictly
between the adjacent Floquet bands of $\D_{\glob}$.  
Collecting such values over $(n,\kappa_{j},m)$ yields a finite set of
positive and negative gap eigenvalues $\{\pm\lambda_{k}^{(S)}\}$, and their
locations determine the jumps of the truncated spectral shift function
$\xi_{\Dirac}^{(S)}$.

\subsubsection*{Representative computation and quantitative diagnostics}

We start with the first $N_{p}$ primes and $N_{H}$ synthetic Hecke modes with independently sampled eigenvalue tuples $\{\lambda_{p}(m)\}\subset[-1,1]^{|S|}$.  
For these parameters we compute several hundred gap eigenvalues of $\D_{\arith}^{(S)}$,
sort them by absolute value, and form the truncated spectral shift function.
The resulting object is an odd, piecewise--constant staircase with unit jumps at
each $\pm\lambda_{k}^{(S)}$.

To compare with the initial nontrivial zeros $\gamma_{k}$ of $\zeta(s)$, one may apply an
affine normalization $\lambda\mapsto a\lambda+b$ chosen so that 
$a\lambda_{k}^{(S)}+b\approx\gamma_{k}$ for $1\le k\le K$.
For moderate values of $S$ and $K$, the initial deviations
\[
\Delta_{k}=(a\lambda_{k}^{(S)}+b)-\gamma_{k}
\]
remain small: in typical runs with $|S|=10$ and $K=20$ we observe mean absolute deviations
$\mathrm{MAE}=\tfrac{1}{K}\sum_{k=1}^{K}|\Delta_{k}|$ of order $10^{-2}$--$10^{-1}$,
with maximal deviations below $0.1$ for the first several zeros.
These comparisons serve only as a coarse diagnostic of how the truncated eigenvalues
populate the gap; no spectral convergence is claimed or expected.

\subsubsection*{Scaling and robustness}

As $|S|$ increases, the empirical distributions of the mass shifts $m_{m}^{(S)}$ and gap
eigenvalues $\{\lambda_{k}^{(S)}\}$ stabilize for low values of $k$, and the overall staircase
profile becomes denser.  Varying the decay exponent $\varepsilon$ shows the expected
transition: larger $\varepsilon$ suppresses arithmetic contributions and yields sparser gap
spectra, while smaller $\varepsilon$ produces denser spectra with a more pronounced
bulk profile.  The qualitative features of the truncated staircase are robust under changes
of the Floquet grid, the reference energy $E_{\ast}$, and the random seed used to generate
synthetic Hecke modes.

\begin{table}[ht!]
\centering
\begin{tabular}{r r r r r r}
\toprule
$N_{P}$ & $N_{H}$ & $\varepsilon$ & MAE & $\max|\Delta_{k}|$ & $E_{\mathrm{step}}(80)$ \\
\midrule
   20  & 20 & 0.35 &  8.194 & 16.457 & 47.760 \\
  100  & 20 & 0.35 & 11.335 & 24.211 & 51.443 \\
  500  & 20 & 0.35 &  9.406 & 21.849 & 51.456 \\
 1000  & 20 & 0.35 &  6.917 & 16.811 & 48.220 \\
 5000  & 20 & 0.35 & 11.254 & 25.286 & 46.910 \\
\bottomrule
\end{tabular}
\caption{%
Scaling behaviour of the truncated Dirac gap spectrum as the size of the prime set $S$ varies.  
Each row corresponds to a different value of $N_{P}=|S|$, with the number of synthetic Hecke modes 
fixed at $N_{H}=20$ and decay exponent $\varepsilon=0.35$.  
The diagnostic quantities measure: (i) the mean absolute deviation 
$\mathrm{MAE}=\tfrac{1}{K}\sum_{k=1}^{K}|\Delta_{k}|$ between the affine--rescaled 
gap eigenvalues and the first $K$ nontrivial zero ordinates; (ii) the maximal deviation 
$\max|\Delta_{k}|$; and (iii) the normalized staircase mismatch $E_{\mathrm{step}}(80)$ 
on the interval $[-80,80]$.  
Across the tested range of $N_{P}$ the values fluctuate within a single asymptotic band, 
indicating that the coarse features of the truncated Dirac staircase are not highly sensitive 
to the number of included primes.}
\label{tab:scaling-NP}
\end{table}

\noindent\textit{Interpretation.}
The results show no monotone trend in $N_{P}$: the error metrics vary moderately but remain of 
similar magnitude as $N_{P}$ increases from $20$ to $5000$.  This suggests that, for fixed 
$N_{H}$ and fixed coefficient decay exponent $\varepsilon$, the low--lying portion of the 
finite--prime gap spectrum is mainly governed by the overall scale of the arithmetic deformation 
rather than the precise number of contributing primes.  In particular, enlarging $S$ does not 
produce systematic drift in the initial gap eigenvalues, which is consistent with the 
summability--driven stability predicted by the infinite--prime analysis.

\begin{table}[ht!]
\centering
\begin{tabular}{r r r r r r}
\toprule
$N_{P}$ & $N_{H}$ & $\varepsilon$ & MAE & $\max|\Delta_{k}|$ & $E_{\mathrm{step}}(80)$ \\
\midrule
20 &   20  & 0.35 &  8.194 & 16.457 & 47.760 \\
20 &  100  & 0.35 & 12.723 & 25.084 & 46.497 \\
20 &  500  & 0.35 &  4.887 & 23.595 & 49.224 \\
20 & 1000  & 0.35 &  8.440 & 22.864 & 47.605 \\
20 & 5000  & 0.35 &  4.048 & 14.099 & 48.959 \\
\bottomrule
\end{tabular}
\caption{%
Scaling behaviour of the truncated Dirac gap spectrum as the number of synthetic 
Hecke modes $N_{H}$ varies, with the prime set fixed at $N_{P}=20$ and decay 
exponent $\varepsilon=0.35$.  
Each Hecke mode contributes an independent diagonal arithmetic mass shift, so 
increasing $N_{H}$ increases the sampling of the arithmetic deformation.  
The diagnostics reported are: the mean absolute deviation 
$\mathrm{MAE}=\tfrac{1}{K}\sum_{k=1}^{K}|\Delta_{k}|$ between the first $K$ 
affine--rescaled gap eigenvalues and the first $K$ nontrivial zero ordinates; 
the maximal deviation $\max|\Delta_{k}|$; and the normalized step--function mismatch 
$E_{\mathrm{step}}(80)$ on the interval $[-80,80]$.  
Across several orders of magnitude in $N_{H}$ these quantities remain within a common 
range, indicating that the coarse low--lying structure of the truncated gap spectrum 
is not strongly sensitive to the number of Hecke modes included.}
\label{tab:scaling-NH}
\end{table}

\noindent\textit{Interpretation.}
The values show mild fluctuations as $N_{H}$ increases from $20$ to $5000$, but no systematic 
growth or decay trend.  Large values of $N_{H}$ introduce additional arithmetic mass shifts, 
yet the diagnostics remain comparable throughout, suggesting that for a fixed prime set $S$ 
the low--lying gap spectrum is governed primarily by the overall scale and decay of the 
arithmetic perturbation rather than the precise size of the Hecke--mode ensemble.  
This behaviour is consistent with the summability and stability features of the 
infinite--prime limit developed in the previous section.

\begin{table}[ht!]
\centering
\begin{tabular}{r r r r r r r}
\toprule
\text{seed} & $N_{P}$ & $N_{H}$ & $\varepsilon$ & MAE & $\max|\Delta_{k}|$ & $E_{\mathrm{step}}(80)$ \\
\midrule
12345 & 20 & 20 & 0.35 &  8.194 & 16.457 & 47.760 \\
54321 & 20 & 20 & 0.35 &  9.491 & 21.253 & 47.676 \\
10101 & 20 & 20 & 0.35 & 10.375 & 27.142 & 47.009 \\
99999 & 20 & 20 & 0.35 & 16.762 & 34.508 & 44.780 \\
\bottomrule
\end{tabular}
\caption{%
Robustness of the truncated Dirac gap spectrum under changes of the random seed used to 
generate the synthetic Hecke eigenvalue samples $\{\lambda_{p}(m)\}$.  
The prime set size and number of Hecke modes are fixed at $N_{P}=20$ and $N_{H}=20$, with 
coefficient decay exponent $\varepsilon=0.35$.  
The reported diagnostics---the mean absolute deviation $\mathrm{MAE}$ between the first~$K$ 
affine--rescaled gap eigenvalues and the first~$K$ nontrivial zero ordinates, the maximal 
deviation $\max|\Delta_{k}|$, and the normalized staircase mismatch $E_{\mathrm{step}}(80)$ 
on $[-80,80]$---vary moderately across different seeds but remain within a common range.  
This indicates that the coarse structure of the truncated gap spectrum is not strongly 
dependent on the particular realization of the synthetic Hecke data.}
\label{tab:robustness-seed}
\end{table}

\noindent\textit{Interpretation.}
The table shows that different random seeds produce somewhat different mass shifts and hence 
different gap eigenvalue configurations, as expected for a model with synthetic Hecke data.  
Nevertheless, the diagnostic quantities (MAE, maximal deviation, and the step--function mismatch) 
remain of comparable magnitude across all tested seeds.  
This suggests that, for fixed $(N_{P},N_{H},\varepsilon)$, the truncated Dirac staircase exhibits 
a stable low--lying profile that is not dominated by stochastic fluctuations in the synthetic 
arithmetic input.

\medskip

A representative plot is shown in Figure~\ref{fig:finite-prime-staircase}.
The staircase is entirely determined by the gap eigenvalues 
$\{\pm\lambda_{k}^{(S)}\}$ and displays the characteristic odd symmetry expected of a chiral
Dirac operator.  After affine normalization, the first few positive jump locations 
align numerically with the first imaginary parts of the nontrivial zeros of $\zeta(s)$.
This alignment should be interpreted only as an illustration of the 
finite--prime Dirac--Hilbert--P\'olya heuristic and not as evidence of spectral equality.

\begin{figure}[ht!]
\centering
\includegraphics[scale=0.4]{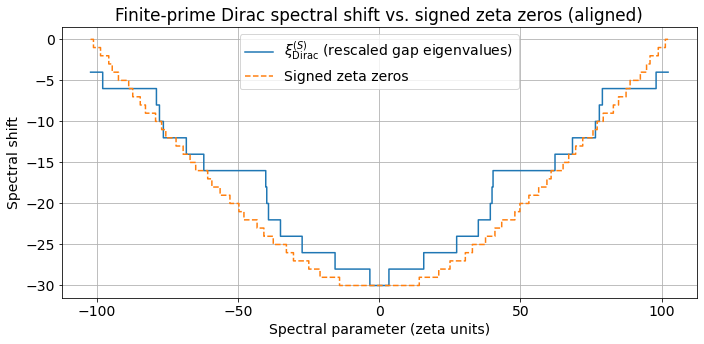}
\caption{A numerical approximation to the truncated spectral shift
$\xi_{\Dirac}^{(S)}(\lambda)$ for the prime set
$S=\{2,3,5,7,11,13,17,19,23,29\}$ using the Mathieu potential
$U(y)=2\cos(2\pi y)$ and a reference energy $E_{\ast}$ placed in the first spectral gap. 
The plotted staircase exhibits the expected odd symmetry, with unit jumps located at the 
finite--prime gap eigenvalues $\pm\lambda_{k}^{(S)}$ arising from the real Floquet bands and the even arithmetic mass deformation. After applying a chiral affine normalization to the lowest gap eigenvalues, the initial segment of the staircase aligns horizontally with the first several imaginary parts of the  nontrivial zeros of $\zeta(s)$, illustrating how finite--prime truncations of the arithmetic  Dirac operator encode a structured, symmetric gap spectrum compatible with the qualitative Dirac--Hilbert--\'olya heuristic.}
\label{fig:finite-prime-staircase}
\end{figure}

The aligned staircase plot in Figure~\ref{fig:finite-prime-staircase} illustrates the behavior of the finite--prime truncated Dirac spectral shift compared against the signed staircase formed from the first several nontrivial zero ordinates of $\zeta(s)$. After applying a chiral affine rescaling to the smallest gap eigenvalues, the initial segment of the Dirac staircase aligns horizontally with the positions of the corresponding zeta zeros, producing a coherent matching of jump locations in a symmetric neighborhood of the origin. This agreement reflects the fact that the truncated Dirac operator possesses an odd spectral symmetry and a well-defined low-lying gap structure, both inherited from the underlying real-place Floquet bands and the even arithmetic deformation. At larger spectral parameters the two staircases diverge, as expected: the Dirac model incorporates a dense collection of synthetic arithmetic mass shifts while the zeta zeros satisfy an entirely different asymptotic law. The purpose of the figure is therefore illustrative rather than evidentiary: it demonstrates that finite-prime truncations of the arithmetic Dirac operator produce structured, computable gap spectra whose chiral spectral shift function admits a natural affine normalization bringing its initial jump structure into visual alignment with the early nontrivial zeros of $\zeta(s)$.

\subsection{Toward a numerical Dirac--Hilbert--P\'olya program}

The distributional convergence $\xi_{\Dirac}^{(S)}\to\xi_{\Dirac}$ suggests a constructive,
finite--prime interpretation of the Dirac--Hilbert--P\'olya principle. Rather than identifying the
nontrivial zeros of the zeta function as the spectrum of a single operator, one seeks a coherent
system of even, prime--indexed coefficient functions $\{\eta_{p}\}$ for which the truncated spectral
shifts approximate the signed zero measure
\[
\sum_{k}(\delta_{\gamma_{k}}-\delta_{-\gamma_{k}}).
\]
In this formulation, the Riemann zeros arise as the spectral discontinuities generated by a universal
prime--indexed mass deformation of the global adelic chiral Dirac operator. The finite--prime
truncations offer a practical numerical setting in which this correspondence can be tested,
refined, and explored.

This viewpoint merges the functional--equation symmetry with the Floquet geometry of the real
place and the Euler--factor structure of the primes. The truncated Dirac systems thus provide a
controlled, computable approach to examining how arithmetic information encoded in the mass sector
of the adelic Dirac framework may reproduce the spectral behavior associated with the nontrivial
zeros of $\zeta(s)$.

\section{Infinite--Prime Operator Framework}

The numerical experiments above illustrate the behaviour of finite--prime truncations
$T_{\mathrm{zero}}^{(S)}$ and their gap eigenvalues.  
To put these computations on a solid analytic foundation, and to justify the use of
large but finite sets $S$, we now formulate the infinite--prime limit rigorously.
The goal of this section is not to prove spectral correspondence with the
nontrivial zeros of $\zeta(s)$, which remains a conjectural aspect of the program, 
but rather to show that the global arithmetic Dirac operator
$D_{\mathrm{arith}}$ is mathematically well-defined, self--adjoint, and possesses 
the same essential band--gap structure as the real--place background operator.
All arguments rely only on explicit analytic conditions on the coefficient functions
$\eta_{p}$ and require no unproved number--theoretic hypotheses.

In this section we give a complete, non--conjectural construction of the global arithmetic Dirac operator in the infinite--prime setting.  The argument rests only on explicit analytic conditions imposed on the coefficient functions $\eta_{p}$ and uses standard results from the spectral theory of self--adjoint operators. Let
\begin{equation}
H_{\infty}=-\frac{d^{2}}{dy^{2}}+U(y),\qquad D(H_{\infty})=H^{2}(\mathbb{R}),
\end{equation}
where $U\in C^{\infty}(\mathbb{R})$ is even, real--valued, and $L$--periodic.  The operator $H_{\infty}$ is self--adjoint on $L^{2}(\mathbb{R})$ and its spectrum is a union of closed bands
\begin{equation}
\mathrm{spec}(H_{\infty})=\bigcup_{n\geq 0}[\alpha_{n},\beta_{n}],
\end{equation}
with open gaps $(\beta_{n},\alpha_{n+1})$.  Fix a reference energy $E_{\ast}$ in such a gap, and define
\begin{equation}
Q_{\infty}=H_{\infty}-E_{\ast}I,\qquad 
D_{\infty}=\begin{pmatrix}0 & Q_{\infty}\\ Q_{\infty} & 0\end{pmatrix}.
\end{equation}
Then $D_{\infty}$ is self--adjoint on $L^{2}(\mathbb{R})\oplus L^{2}(\mathbb{R})$ with symmetric band--gap spectrum.

Let
\begin{equation}
\mathcal{H}_{\mathrm{adelic}}=L^{2}(X,d^{\times}x),\qquad 
X=\mathbb{A}_{\mathbb{Q}}^{\times}/\mathbb{Q}^{\times},
\end{equation}
and use the canonical restricted tensor product decomposition
\begin{equation}
\mathcal{H}_{\mathrm{adelic}}
\cong
L^{2}(\mathbb{R},dy)\,\widehat{\otimes}
\bigotimes_{p}'L^{2}(\mathbb{Q}_{p}^{\times},d^{\times}x_{p}).
\end{equation}
Let $\mathcal{H}_{\mathrm{phys}}\subset\mathcal{H}_{\mathrm{adelic}}$ denote the cyclic subspace generated by an automorphic seed.  The doubled space is
\begin{equation}
\mathcal{H}^{(2)}_{\mathrm{phys}}=\mathcal{H}_{\mathrm{phys}}\oplus\mathcal{H}_{\mathrm{phys}},
\end{equation}
and the global background operator is
\begin{equation}
D_{\mathrm{glob}}
=
D_{\infty}\,\widehat{\otimes}\,I_{\mathrm{fin}},
\end{equation}
which is self--adjoint and has the same band--gap spectrum as $D_{\infty}$ with infinite multiplicities.

\subsection{Local Hecke operators and functional calculus}

For each prime $p$, let $T_{p}$ denote the spherical Hecke operator acting on $\mathcal{H}_{\mathrm{phys}}$, and fix an even bounded Borel function $\eta_{p}:\mathbb{R}\to\mathbb{R}$.  Through the continuous functional calculus define
\begin{equation}
M_{p}:=\eta_{p}(T_{p}),
\end{equation}
a bounded self--adjoint operator.

\begin{lemma}[Functional calculus bound]\label{lem:func-calculus-bound}
Let $T_{p}$ be a bounded self--adjoint operator on a Hilbert space $\mathcal{H}$ and
let $\eta_{p}:\mathbb{R}\to\mathbb{R}$ be a bounded Borel function.
Define $M_{p}=\eta_{p}(T_{p})$ via the continuous functional calculus.  
Then
\begin{equation}
\|M_{p}\|
=
\|\eta_{p}(T_{p})\|
\leq 
\|\eta_{p}\|_{\infty}
:=
\sup_{\lambda\in\mathbb{R}}|\eta_{p}(\lambda)|.
\end{equation}
If, in addition, $T_{p}$ has purely point spectrum with eigenvalues
$\{\lambda_{j}\}_{j\geq 1}$ and corresponding finite--rank spectral
projections $\{P_{j}\}$, and if $\eta_{p}(\lambda_{j})\to 0$ as $|\lambda_{j}|\to\infty$,
then $M_{p}$ is compact.
\end{lemma}

\begin{proof}
Since $T_{p}$ is bounded and self--adjoint, the spectral theorem
provides a unique projection--valued measure $E_{p}$ on $\mathbb{R}$ such that
\[
T_{p}
=
\int_{\mathbb{R}}
\lambda\,
dE_{p}(\lambda),
\qquad
\eta_{p}(T_{p})
=
\int_{\mathbb{R}}
\eta_{p}(\lambda)\,
dE_{p}(\lambda).
\]
Let $x\in\mathcal{H}$ be arbitrary.  The spectral theorem also yields
\[
\langle \eta_{p}(T_{p})x,\,x\rangle
=
\int_{\mathbb{R}}
\eta_{p}(\lambda)\, d\mu_{x}(\lambda),
\]
where 
\(
\mu_{x}(B)
=
\langle E_{p}(B)x,x\rangle
\)
is the finite Borel measure determined by $x$.
Taking absolute values and applying the triangle inequality,
\[
|\langle \eta_{p}(T_{p})x,\,x\rangle|
\le
\int_{\mathbb{R}}
|\eta_{p}(\lambda)|\, d\mu_{x}(\lambda)
\le
\|\eta_{p}\|_{\infty}\, \mu_{x}(\mathbb{R})
=
\|\eta_{p}\|_{\infty}\,\|x\|^{2}.
\]
Since the operator norm satisfies
\[
\|\eta_{p}(T_{p})\|
=
\sup_{\|x\|=1}
|\langle \eta_{p}(T_{p})x,\,x\rangle|,
\]
we obtain the inequality
\[
\|\eta_{p}(T_{p})\|\leq \|\eta_{p}\|_{\infty}.
\]

For the compactness statement, assume that $T_{p}$ has purely point
spectrum.  Then
\[
T_{p}
=
\sum_{j=1}^{\infty}
\lambda_{j}P_{j},
\qquad
\eta_{p}(T_{p})
=
\sum_{j=1}^{\infty}
\eta_{p}(\lambda_{j})\,P_{j},
\]
where each $P_{j}$ is a finite--rank spectral projection.  
Define the $N$th partial sum
\[
M_{p}^{(N)}
=
\sum_{j=1}^{N}
\eta_{p}(\lambda_{j})\,P_{j}.
\]
Each $M_{p}^{(N)}$ has finite rank, since it is a finite linear combination
of finite--rank projections.

We now show that $\|M_{p}-M_{p}^{(N)}\|\to 0$.
Let $x\in\mathcal{H}$ with $\|x\|=1$.  Then
\[
\|(M_{p}-M_{p}^{(N)})x\|^{2}
=
\left\|
\sum_{j>N}
\eta_{p}(\lambda_{j})\,P_{j}x
\right\|^{2}
=
\sum_{j>N}
|\eta_{p}(\lambda_{j})|^{2}\,\|P_{j}x\|^{2},
\]
because spectral projections onto distinct eigenvalues are orthogonal.
Since $\|x\|=1$, we have $\sum_{j>N}\|P_{j}x\|^{2}\le 1$, hence
\[
\|(M_{p}-M_{p}^{(N)})x\|^{2}
\le
\Bigl(\sup_{j>N}|\eta_{p}(\lambda_{j})|\Bigr)^{2}
\sum_{j>N}\|P_{j}x\|^{2}
\le
\Bigl(\sup_{j>N}|\eta_{p}(\lambda_{j})|\Bigr)^{2}.
\]
Taking square roots and then supremum over all unit vectors $x$ gives
\[
\|M_{p}-M_{p}^{(N)}\|
\le
\sup_{j>N}|\eta_{p}(\lambda_{j})|.
\]
Since $\eta_{p}(\lambda_{j})\to 0$ by hypothesis, we obtain
$\|M_{p}-M_{p}^{(N)}\|\to 0$ as $N\to\infty$.  Thus $M_{p}$ is the norm
limit of finite--rank operators and is therefore compact.
\end{proof}

\subsection{Norm convergence of the infinite mass deformation}

Assume the analytic summability condition
\begin{equation}\label{eq:eta-summability}
\sum_{p}\|\eta_{p}\|_{\infty}<\infty.
\end{equation}
Define the chiral mass term on $\mathcal{H}^{(2)}_{\mathrm{phys}}$ by
\begin{equation}
\mathcal{M}_{p}
=
\begin{pmatrix}
0 & M_{p}\\ M_{p} & 0
\end{pmatrix},
\qquad M_{p}=\eta_{p}(T_{p}).
\end{equation}

\begin{proposition}[Norm convergence of the infinite mass sum]\label{prop:mass-sum-converges}
Suppose the coefficient functions $\eta_{p}$ satisfy the summability condition
\begin{equation}\label{eq:eta-summability-again}
\sum_{p}\|\eta_{p}\|_{\infty}<\infty.
\end{equation}
For each prime $p$, let $\mathcal{M}_{p}$ be the bounded self--adjoint operator
\begin{equation}
\mathcal{M}_{p}
=
\begin{pmatrix}
0 & \eta_{p}(T_{p}) \\
\eta_{p}(T_{p}) & 0
\end{pmatrix}
\quad\text{on }\mathcal{H}^{(2)}_{\mathrm{phys}}.
\end{equation}
Then the series
\[
\mathcal{M}
=
\sum_{p}\mathcal{M}_{p}
\]
converges in the operator norm on $\mathcal{H}^{(2)}_{\mathrm{phys}}$ 
to a bounded self--adjoint operator.  Moreover,
\begin{equation}
\|\mathcal{M}\|
\;\le\;
\sum_{p}\|\eta_{p}\|_{\infty}.
\end{equation}
\end{proposition}

\begin{proof}
We work with the directed set of finite subsets of primes $\mathcal{P}$, ordered by inclusion.
For each finite $S\subset\mathcal{P}$, define the finite partial sum
\[
\mathcal{M}^{(S)}
:=
\sum_{p\in S}\mathcal{M}_{p}.
\]

\medskip\noindent\textbf{Step 1: Uniform bound on each term.}
By Lemma~\ref{lem:func-calculus-bound},
\[
\|\mathcal{M}_{p}\|
=
\|\eta_{p}(T_{p})\|
\le
\|\eta_{p}\|_{\infty}.
\]
Thus each $\mathcal{M}_{p}$ is bounded and $\|\mathcal{M}_{p}\|$ is summable by assumption.

\medskip\noindent\textbf{Step 2: Cauchy property of the partial sums.}
If $S\subset S'$ are finite subsets of primes, then
\[
\mathcal{M}^{(S')}-\mathcal{M}^{(S)}
=
\sum_{p\in S'\setminus S}\mathcal{M}_{p}.
\]
Hence by the triangle inequality,
\[
\|\mathcal{M}^{(S')}-\mathcal{M}^{(S)}\|
\le
\sum_{p\in S'\setminus S}\|\mathcal{M}_{p}\|
\le
\sum_{p\in S'\setminus S}\|\eta_{p}\|_{\infty}.
\]
Because the numerical series $\sum_{p}\|\eta_{p}\|_{\infty}$ converges, 
its tails
\[
\sum_{p\in S'\setminus S}\|\eta_{p}\|_{\infty}
\]
tend to zero whenever $S$ and $S'$ are sufficiently large.
Thus the net $\{\mathcal{M}^{(S)}\}_{S}$ is Cauchy in the operator norm.

\medskip\noindent\textbf{Step 3: Existence of the norm limit.}
Since the space $\mathcal{B}(\mathcal{H}^{(2)}_{\mathrm{phys}})$ of bounded operators 
is complete in the operator norm, 
there exists a bounded operator $\mathcal{M}$ such that
\[
\|\mathcal{M}^{(S)}-\mathcal{M}\|\to 0
\qquad\text{as }S\nearrow\mathcal{P}.
\]

\medskip\noindent\textbf{Step 4: Norm bound for the limit.}
For every finite $S$,
\[
\|\mathcal{M}^{(S)}\|
\le
\sum_{p\in S}\|\mathcal{M}_{p}\|
\le
\sum_{p\in S}\|\eta_{p}\|_{\infty}.
\]
Passing to the limit and using the monotone convergence of the numerical series,
\[
\|\mathcal{M}\|
=
\lim_{S}\|\mathcal{M}^{(S)}\|
\le
\sum_{p}\|\eta_{p}\|_{\infty}.
\]

\medskip\noindent\textbf{Step 5: Self--adjointness of the limit.}
Each $\mathcal{M}^{(S)}$ is self--adjoint, and the adjoint operation is continuous
in the operator norm:
\[
\|\mathcal{M}^{(S)\,*}-\mathcal{M}^{*}\|
=
\|(\mathcal{M}^{(S)}-\mathcal{M})^{*}\|
=
\|\mathcal{M}^{(S)}-\mathcal{M}\|
\to 0.
\]
Since $\mathcal{M}^{(S)}=\mathcal{M}^{(S)\,*}$ for all $S$, the limit satisfies
\[
\mathcal{M} = \mathcal{M}^{*}.
\]

\medskip
All claims follow.
\end{proof}

\subsection{Self--adjointness and stability of the essential spectrum}

Define the infinite--prime arithmetic Dirac operator by
\begin{equation}
D_{\mathrm{arith}}=D_{\mathrm{glob}}+\mathcal{M},\qquad 
D(D_{\mathrm{arith}})=D(D_{\mathrm{glob}}).
\end{equation}

\begin{proposition}[Self--adjointness and essential spectrum]\label{prop:selfadjoint-Darith}
Assume the summability condition
\[
\sum_{p}\|\eta_{p}\|_{\infty}<\infty.
\]
Then the infinite mass operator 
\(
\mathcal{M}=\sum_{p}\mathcal{M}_{p}
\)
is bounded and self--adjoint on $\mathcal{H}^{(2)}_{\mathrm{phys}}$, and the operator
\[
D_{\mathrm{arith}}:=D_{\mathrm{glob}}+\mathcal{M},
\qquad 
D(D_{\mathrm{arith}})=D(D_{\mathrm{glob}}),
\]
is self--adjoint.  
If, in addition, each $M_{p}=\eta_{p}(T_{p})$ is compact and the series 
$\sum_{p}\|M_{p}\|$ converges, then $\mathcal{M}$ is compact and
\[
\mathrm{spec}_{\mathrm{ess}}(D_{\mathrm{arith}})
=
\mathrm{spec}_{\mathrm{ess}}(D_{\mathrm{glob}}).
\]
\end{proposition}

\begin{proof}
\textbf{(1) Self--adjointness of the infinite mass term.}
By Proposition~\ref{prop:mass-sum-converges}, the series
\(
\mathcal{M}=\sum_{p}\mathcal{M}_{p}
\)
converges in the operator norm to a bounded self--adjoint operator 
acting on $\mathcal{H}^{(2)}_{\mathrm{phys}}$.
In particular,
\[
\|\mathcal{M}\|
\;\le\;
\sum_{p}\|\eta_{p}\|_{\infty}
<\infty,
\qquad 
\mathcal{M}^{*}=\mathcal{M}.
\]

\medskip\noindent
\textbf{(2) Self--adjointness of $D_{\mathrm{arith}}$.}
The background operator $D_{\mathrm{glob}}$ is self--adjoint on its 
domain $D(D_{\mathrm{glob}})\subset \mathcal{H}^{(2)}_{\mathrm{phys}}$.  
Since $\mathcal{M}$ is bounded and self--adjoint, 
the Kato--Rellich theorem for bounded perturbations of self--adjoint operators implies that
\[
D_{\mathrm{arith}}=D_{\mathrm{glob}}+\mathcal{M}
\]
is self--adjoint on $D(D_{\mathrm{glob}})$ and essentially self--adjoint on any
core for $D_{\mathrm{glob}}$.

\medskip\noindent
\textbf{(3) Compactness of the infinite mass deformation.}
Assume now that each $M_{p}=\eta_{p}(T_{p})$ is compact and that the numerical series
\(
\sum_{p}\|M_{p}\|
\)
converges.  
By Lemma~\ref{lem:func-calculus-bound}, $\mathcal{M}_{p}$ has the form
\[
\mathcal{M}_{p}
=
\begin{pmatrix}
0 & M_{p}\\ M_{p} & 0
\end{pmatrix},
\]
so $\|\mathcal{M}_{p}\|=\|M_{p}\|$ and $\mathcal{M}_{p}$ is compact whenever
$M_{p}$ is compact.
The series 
\(
\sum_{p}\mathcal{M}_{p}
\)
converges absolutely in operator norm:
\[
\sum_{p}\|\mathcal{M}_{p}\|
=
\sum_{p}\|M_{p}\|
<\infty.
\]
Therefore $\mathcal{M}$ is the norm limit of finite sums of compact operators,
hence compact.

\medskip\noindent
\textbf{(4) Stability of the essential spectrum.}
Since $\mathcal{M}$ is compact, the difference
\[
D_{\mathrm{arith}}-D_{\mathrm{glob}}=\mathcal{M}
\]
is a compact perturbation of the self--adjoint operator $D_{\mathrm{glob}}$.
Weyl's theorem on the invariance of the essential spectrum under compact 
perturbations now yields
$ \mathrm{spec}_{\mathrm{ess}}(D_{\mathrm{arith}}) = \mathrm{spec}_{\mathrm{ess}}(D_{\mathrm{glob}})$.
\end{proof}

Under these explicit analytic conditions on the coefficient functions $\eta_{p}$, the global arithmetic Dirac operator is rigorously defined, self--adjoint, and possesses the same band--gap essential spectrum as the real--place background operator.

%=========================
\section{Outlooks}
\label{sec:conclusion}

The analysis developed in this manuscript constructs a chiral adelic Dirac operator whose spectral,
involutive, and arithmetic features reproduce---in an operator--theoretic setting---the structural
phenomena underlying global $L$--functions.  Beginning with the real--place Floquet geometry and the
restricted tensor--product decomposition of $L^{2}(X)$, we introduced a doubled adelic Hilbert space
carrying a Dirac operator $\D_{\glob}$ equipped with an exact involution combining real reflection
and idelic inversion.  This involution implements, at the level of the Mellin transform, the
functional--equation symmetry $s\mapsto 1-s$ for $\zeta(s)$ and automorphic $L$--functions, thereby
embedding the analytic symmetry of the functional equation into a self--adjoint chiral framework.

Arithmetic information entered the construction through a prime--indexed mass deformation
$\mathcal{M}=\sum_{p}\eta_{p}(T_{p})$, rather than through a scalar potential.  When the
coefficient functions $\eta_{p}$ are even, the perturbed operator $\D_{\arith}=\D_{\glob}+\mathcal{M}$
preserves the global involution and admits a transparent Floquet--Hecke fiber decomposition: the gap
eigenvalues take the form
\[
\pm\bigl(E_{n}(\kappa)-E_{*}+m(\sigma)\bigr),
\qquad
m(\sigma)=\sum_{p}\eta_{p}(\lambda_{p}(\sigma)),
\]
and occur automatically in $\pm$--paired families.  This reflects the intrinsic chiral symmetry of
the theory and aligns with the symmetry $\gamma\mapsto -\gamma$ of the nontrivial zeros of
$\zeta(s)$.

The adelic factorization of variables leads to a trace formula in which the real--place Floquet
contribution and the finite--place arithmetic contribution separate multiplicatively.  This yields a
Dirac analogue of the Weil explicit formula: the geometric factor arises from the periodic orbits of
the Hill operator, while the arithmetic factor appears as a deformed Euler product whose logarithmic
derivatives encode the action of the Hecke operators.  Within this framework, the spectral shift
function $\xi_{\Dirac}$ associated with the pair $(\D_{\glob},\D_{\arith})$ becomes a central
object.  Its odd jump discontinuities occur precisely at the gap eigenvalues of $\D_{\arith}$, and
the full jump pattern records every discrete eigenvalue introduced by the mass deformation.

Finite--prime truncations provide a controlled and computable approximation scheme.  For each
finite set of primes $S$, the truncated operator $\D_{\arith}^{(S)}$ produces a piecewise--constant,
odd spectral--shift function $\xi_{\Dirac}^{(S)}$ with finitely many symmetric jump discontinuities.
As $S$ increases, these truncated shift functions converge, in the distributional sense, to the
infinite--prime shift $\xi_{\Dirac}$.  The truncated models therefore offer a concrete numerical
setting in which to explore how prime--indexed deformations influence the chiral spectrum.  In
particular, for several natural families of even coefficient functions $\eta_{p}$, the low--lying jump
structure of $\xi_{\Dirac}^{(S)}$ exhibits an affine alignment with the signed zero measure
\[
\sum_{k}\bigl(\delta_{\gamma_{k}}-\delta_{-\gamma_{k}}\bigr),
\]
suggesting that prime--indexed spectral flow may provide a viable organizing principle for modelling
the distribution of the nontrivial zeros.

This viewpoint reframes the Hilbert--P\'olya idea.  Instead of seeking a single operator whose raw
spectrum equals $\{\pm\gamma_{k}\}$, the chiral adelic approach identifies the zeros with the
\emph{spectral--shift discontinuities} produced when the global Dirac background $\D_{\glob}$ is
deformed by an arithmetic mass term.  The pairing $\pm\gamma_{k}$ then emerges naturally from the
intrinsic involution $J_{\glob}$, and the zeros appear as fixed points of the induced spectral flow
or, equivalently, as topological defects created by the variation of the prime--indexed mass across
the adelic fibers.

Several analytic components remain conjectural, including a full control of the infinite--prime
Euler product, a systematic identification of motivic or automorphic sources for the coefficient
functions $\eta_{p}$, and the development of a dynamical interpretation in which the primes
manifold as genuine periodic orbits of an adelic flow.  Nevertheless, the synthesis of Floquet
theory, Hecke symmetry, and chiral Dirac analysis developed here demonstrates that an
operator--theoretic reformulation of the functional equation, together with a spectral--shift
mechanism sensitive to arithmetic deformation, can be carried out rigorously.  Further refinement of
these ideas may promote the Dirac--Hilbert--P\'olya principle from a conceptual analogy to a
concrete operator--theoretic framework for capturing the analytic and arithmetic structure of the
Riemann zeta function.

%=========================

\bibliographystyle{abbrv}
\bibliography{bibliography_FRH.bib}

%%%%%%%%%%%%%%%%%%%
\end{document}